\documentclass[11pt]{article}
\usepackage{amsmath,amssymb,amsthm}
\usepackage{aliascnt}
\usepackage{graphicx}
\usepackage{fullpage,hyperref,color,cleveref,booktabs,array,enumitem,placeins}

\usepackage[linesnumbered,ruled,vlined]{algorithm2e}

\allowdisplaybreaks

\newcommand{\Expose}{\operatorname{Expose}}
\providecommand{\TV}{\operatorname{TV}}
\providecommand{\Tr}{\operatorname{Tr}}
\providecommand{\Unif}{\operatorname{Unif}}

\newtheorem{theorem}{Theorem}[section]
\newaliascnt{lemma}{theorem}
\newtheorem{lemma}[lemma]{Lemma}
\aliascntresetthe{lemma}
\newaliascnt{claim}{theorem}

\aliascntresetthe{claim}
\newaliascnt{corollary}{theorem}
\newtheorem{corollary}[corollary]{Corollary}
\aliascntresetthe{corollary}
\newaliascnt{proposition}{theorem}

\aliascntresetthe{proposition}
\theoremstyle{definition}
\newaliascnt{definition}{theorem}
\newtheorem{definition}[definition]{Definition}
\aliascntresetthe{definition}
\newaliascnt{example}{theorem}

\aliascntresetthe{example}
\theoremstyle{remark}
\newaliascnt{remark}{theorem}

\aliascntresetthe{remark}
\newaliascnt{assumption}{theorem}

\aliascntresetthe{assumption}

\crefname{theorem}{theorem}{theorems}
\Crefname{theorem}{Theorem}{Theorems}
\crefname{lemma}{lemma}{lemmas}
\Crefname{lemma}{Lemma}{Lemmas}
\crefname{claim}{claim}{claims}
\Crefname{claim}{Claim}{Claims}
\crefname{corollary}{corollary}{corollaries}
\Crefname{corollary}{Corollary}{Corollaries}
\crefname{proposition}{proposition}{propositions}
\Crefname{proposition}{Proposition}{Propositions}
\crefname{definition}{definition}{definitions}
\Crefname{definition}{Definition}{Definitions}
\crefname{example}{example}{examples}
\Crefname{example}{Example}{Examples}
\crefname{remark}{remark}{remarks}
\Crefname{remark}{Remark}{Remarks}
\crefname{assumption}{assumption}{assumptions}
\Crefname{assumption}{Assumption}{Assumptions}
\crefname{algocf}{algorithm}{algorithms}
\Crefname{algocf}{Algorithm}{Algorithms}

\title{Testing Monotonicity of Real-Valued Functions on DAGs}

\begin{document}
\hypersetup{pageanchor=false}

\author{Yuichi Yoshida \\
National Institute of Informatics\\
\texttt{yyoshida@nii.ac.jp}}
\maketitle

\begin{abstract}
    We study monotonicity testing of real-valued functions on directed acyclic graphs (DAGs) with $n$ vertices.
    Let $m$ and $\ell$ be the numbers of edges in the transitive reduction and the transitive
    closure, respectively.  For $1\le c\le d\le2$, define
    \[
      u(c,d):=\min\left\{\frac12,\frac c3,\frac{c+d}{2}-1\right\}.
    \]
    We show that every family of DAGs with $m=n^{c+o(1)}$ and $\ell=n^{d+o(1)}$ admits, for every
    fixed $\varepsilon\in(0,1)$, a non-adaptive tester with one-sided error that uses
    $O_\varepsilon(n^{u(c,d)+o(1)})$ queries.
    Conversely, we show that for every sufficiently small fixed $\varepsilon>0$ and every fixed
    $(c,d)$, there are families of DAGs satisfying $m=n^{c+o(1)}$ and $\ell=n^{d+o(1)}$ on which
    every randomized non-adaptive tester, even with two-sided error, requires
    $n^{u(c,d)-o(1)}$ queries, making the upper bound tight up to a factor $n^{o(1)}$.
    Our main technical contribution is a lower-bound technique based on Ruzsa--Szemer\'edi families
    of positive matchings.
\end{abstract}

\thispagestyle{empty}
\newpage
\hypersetup{pageanchor=true}
\pagenumbering{roman}
\thispagestyle{empty}
\setcounter{tocdepth}{2}
\tableofcontents
\clearpage
\pagenumbering{arabic}

\section{Introduction}
Testing monotonicity on structured domains is a central problem in property testing and a standard example
in the sublinear-time model; see, e.g.,~\cite{EKKRV00,Goldreich17,BhattacharyyaYoshida22,GGR98,Ron09}.
It was first studied on the Boolean hypercube and product domains and has since been extended to richer
product structures, general posets, and families of graphs~\cite{AC06,BGJRW12,CS13a,DGLRRS99,fischer2002monotonicity,GGLRS00,HK08,Raskhodnikova10}.
One way to model posets is via reachability in directed acyclic graphs (DAGs).  In this work we study
the \emph{massively parameterized model}~\cite{Ron09,Goldreich17}, where the entire DAG is given explicitly.

\subsection{Problem and background}
For every positive integer $K$, write $[K]:=\{1,\ldots,K\}$.
Let $G=(V,E)$ be a DAG with $n=|V|$ vertices.
A function $f:V\to\mathbb{R}$ is monotone if $f(u)\le f(v)$ whenever $u$ reaches $v$ in $G$
(equivalently, whenever $(u,v)$ is an edge of the transitive closure).
Given full access to $G$, a proximity parameter $\varepsilon>0$, and oracle access to $f$,
the goal is to distinguish monotone functions from those that are $\varepsilon$-far from
monotonicity with constant success probability.
We do not charge for preprocessing or computations on $G$; only queries to $f$ are counted.
A query to a vertex $v$ returns its full value $f(v)$.
A tester is \emph{non-adaptive} if all queries are fixed in advance and has \emph{one-sided error}
if it always accepts monotone functions.

For general posets, the classical tester of Fischer et al.~\cite{fischer2002monotonicity}, which is
non-adaptive and has one-sided error,
achieves
\[
O\!\left(\min\left\{\sqrt{\frac{n}{\varepsilon}},\ \frac{\ell}{\varepsilon n}\right\}\right)
\]
queries, where $\ell$ is the number of comparable pairs, i.e., the number of edges in the transitive
closure of $G$.
They also give an $n^{\Omega(1/\log\log n)}$ lower bound for Boolean ranges against non-adaptive
testers with two-sided error, using constructions based on Ruzsa--Szemer\'edi graphs~\cite{RS78}.
This leaves a large gap from the upper bound $O(\sqrt{n/\varepsilon})$.

\subsection{Our results}
\paragraph{Query-complexity phase diagram.}
Our main result determines, up to a factor $n^{o(1)}$, how the worst-case query complexity of
randomized non-adaptive testers depends on the sizes of the transitive reduction and transitive
closure of the DAG.
We write $m$ for the number of edges in the transitive reduction of $G$ and $\ell$ for the number
of edges in its transitive closure.

\begin{theorem}[Query complexity for every pair of reachability exponents]
\label{thm:query-complexity-phase-diagram}
Fix constants $1\le c\le d\le2$ and define
\[
  u(c,d):=
  \min\left\{\frac12,\frac c3,\frac{c+d}{2}-1\right\}.
\]
\begin{enumerate}[label=(\roman*)]
  \item For every fixed $\varepsilon\in(0,1)$, every family of $n$-vertex DAGs satisfying
  $m=n^{c+o(1)}$ and $\ell=n^{d+o(1)}$ admits a non-adaptive one-sided $\varepsilon$-tester using
  $O_\varepsilon(n^{u(c,d)+o(1)})$ queries.
  \item For every sufficiently small fixed $\varepsilon>0$, there is a family of $n$-vertex DAGs
  satisfying $m=n^{c+o(1)}$ and $\ell=n^{d+o(1)}$ on which every randomized non-adaptive
  $\varepsilon$-tester, even with two-sided error, requires $n^{u(c,d)-o(1)}$ queries.
\end{enumerate}
\end{theorem}

\Cref{fig:monotonicity-tester} shows the resulting phase diagram, and
\Cref{tab:main-results} summarizes the underlying quantitative upper and lower bounds, including
their dependence on $\varepsilon$.

\begin{figure}[t]
    \centering
    \includegraphics[width=0.4\linewidth]{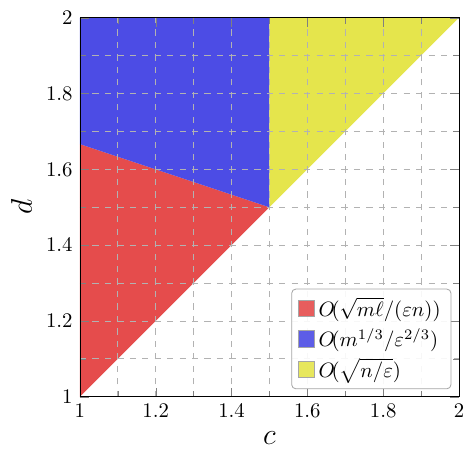}
    \caption{For fixed $\varepsilon>0$, the colors indicate which of the three terms in the
    upper bound $u(c,d)=\min\{1/2,c/3,(c+d)/2-1\}$ determines the exponent of the query
    complexity when
    $m=n^{c+o(1)}$ and $\ell=n^{d+o(1)}$
    (\Cref{thm:query-complexity-phase-diagram}\textup{(i)}).  The displayed triangle
    $1\le c\le d\le2$ is the natural parameter range for connected underlying undirected graphs;
    the bounds themselves do not require connectivity.  For every sufficiently small fixed
    $\varepsilon$, \Cref{thm:query-complexity-phase-diagram}\textup{(ii)} provides a matching lower
    bound at every fixed pair in the triangle, up to an arbitrarily small exponent loss, against
    randomized non-adaptive testers with two-sided error.}
    \label{fig:monotonicity-tester}
\end{figure}

\paragraph{Upper bounds.}
The upper bound in \Cref{thm:query-complexity-phase-diagram}\textup{(i)} combines the classical
$O(\sqrt{n/\varepsilon})$-query tester~\cite{fischer2002monotonicity} with two new
testers.  Our $m$-dependent tester combines samples from the transitive reduction with two
independent samples of vertices and uses
\[
  O\!\left(\frac{m^{1/3}}{\varepsilon^{2/3}}\right)
\]
queries (\Cref{thm:m13}).  Our $(m,\ell)$-dependent tester samples from both the transitive
reduction and the transitive closure and uses
\[
  O\!\left(\frac{\sqrt{m\ell}}{\varepsilon n}\right)
\]
queries (\Cref{thm:sqrt-ml}).  Both testers are non-adaptive and have one-sided error.  Taking the
best of the three gives
\[
  O\!\left(
    \min\left\{
      \sqrt{\frac n\varepsilon},
      \frac{m^{1/3}}{\varepsilon^{2/3}},
      \frac{\sqrt{m\ell}}{\varepsilon n}
    \right\}
  \right)
\]
queries.  For fixed $\varepsilon$, $m=n^{c+o(1)}$, and $\ell=n^{d+o(1)}$, the three terms have
exponents $1/2$, $c/3$, and $(c+d)/2-1$, respectively, which gives $u(c,d)$.  The two new bounds
improve the classical bound when $m=o(n^{3/2})$ and $m\ell=o(n^3)$, respectively.

\paragraph{Lower bounds via positive matchings.}
The lower bound in \Cref{thm:query-complexity-phase-diagram}\textup{(ii)} matches all three terms in the upper
bound, up to a factor $n^{o(1)}$, even for randomized non-adaptive testers with two-sided error.
Our construction begins with a \emph{PMRS family}, short for a Ruzsa--Szemer\'edi family of positive
matchings.  A matching $M$ in a graph is \emph{positive} if some weighting of the vertices makes the sum of
the weights of the endpoints positive exactly on the edges of $M$.  For a bipartite DAG, this is equivalent
to $M$ being exactly the set of violating edges of a real-valued function.

For every constant $\delta>0$, every fixed $\varepsilon\in(0,1/8]$, and infinitely many $n$, PMRS
families yield $n$-vertex bipartite DAGs on which every randomized non-adaptive $\varepsilon$-tester,
even with two-sided error, needs
\[
  \Omega_\delta\!\left(\frac{n^{1/2-\delta}}{\sqrt{\varepsilon}}\right)
\]
queries (\Cref{thm:pmrs_near_sqrt}).  This nearly matches the classical
$O(\sqrt{n/\varepsilon})$ upper bound and improves the previous
$n^{\Omega(1/\log\log n)}$ lower bound~\cite{fischer2002monotonicity}.

To obtain matching lower bounds for the two parameter-dependent terms, we construct PMRS families with
\emph{bounded label exposure}: no small query set contains both endpoints of edges from more than a small fraction of
the designated matchings.  On suitable families of DAGs, this gives
\[
  \Omega_\delta\!\left(\frac{m^{1/3-\delta}}{\varepsilon^{2/3}}\right)
  \qquad\text{and}\qquad
  \Omega\!\left(\frac{1}{\log n}\cdot
  \frac{\sqrt{m\ell}}{\varepsilon n}\right)
\]
for every $\delta>0$ and every sufficiently small fixed $\varepsilon>0$
(\Cref{thm:m-third-twosided,thm:ell-dependent-twosided}).  The first differs from its upper bound
by an arbitrarily small power of $m$, and the second by a logarithmic factor.

\begin{table}[t!]
  \centering
  \caption{Quantitative upper bounds and lower bounds retaining their dependence on $\varepsilon$.
  The lower bounds are attained on suitable families of DAGs.}
  \label{tab:main-results}
  \small
  \begin{tabular}{@{}p{0.14\linewidth}p{0.38\linewidth}>{\raggedright\arraybackslash}p{0.42\linewidth}@{}}
\toprule
Parameters & Upper bound (one-sided) & Lower bound on suitable DAGs (two-sided) \\
\midrule
$n$ & $O(\sqrt{n/\varepsilon})$~\cite{fischer2002monotonicity}
    & $\Omega_\delta(n^{1/2-\delta}/\sqrt{\varepsilon})$ (\Cref{thm:pmrs_near_sqrt}) \\
\addlinespace
$m$ & $O(m^{1/3}/\varepsilon^{2/3})$ (\Cref{thm:m13})
    & $\Omega_\delta(m^{1/3-\delta}/\varepsilon^{2/3})$ (\Cref{thm:m-third-twosided}) \\
\addlinespace
$m,\ell$ & $O(\sqrt{m\ell}/(\varepsilon n))$ (\Cref{thm:sqrt-ml})
    & $\Omega(\sqrt{m\ell}/(\varepsilon n\log n))$ (\Cref{thm:ell-dependent-twosided}) \\
\bottomrule
  \end{tabular}
\end{table}

\FloatBarrier

\subsection{Technical overview}
\subsubsection{Upper bounds parameterized by the transitive reduction and closure}
For fixed $\varepsilon$, the three terms $1/2$, $c/3$, and $(c+d)/2-1$ in $u(c,d)$ come from
the classical tester, our tester whose query bound depends on $m$ but not $\ell$, and our tester
using both the transitive reduction and closure, respectively.  We describe the two terms supplied
by our new testers.

\paragraph{Partitioning by the first violating edge.}
The upper bounds start from the characterization of distance via the violation graph
(\Cref{sec:preliminaries}).  If $f$ is $\varepsilon$-far, the violation graph in the transitive
closure contains a matching $M$ of size greater than $\varepsilon n$.  Thus there are many disjoint
comparable pairs $(u,v)$ with $u\leadsto v$ and $f(u)>f(v)$.  The difficulty is that these pairs
may be edges only of the transitive closure, whereas the transitive reduction may contain very few
violating edges.

For every $(u,v)\in M$, fix a path $P_{uv}$ in the transitive reduction and assign $(u,v)$ to the
first violating edge on this path.  Writing
\[
  F_f:=\{(x,y)\in E:f(x)>f(y)\},
  \qquad
  M=\biguplus_{e\in F_f}M_e,
\]
produces a bucket $M_e$ for each violating edge $e\in E$.  Because $M$ is a matching, the sets
$A_e$ of sources and $B_e$ of sinks appearing in the buckets are disjoint across distinct $e$.
For every nonempty bucket, split its sources and sinks around a median value to obtain sets
$A_e^+\subseteq A_e$ and $B_e^-\subseteq B_e$, each of size at least $|M_e|/2$.
Every $a\in A_e^+$ reaches the tail of $e$, and the head of $e$ reaches every $b\in B_e^-$;
moreover, $f(a)>f(b)$.  Thus every pair in $A_e^+\times B_e^-$ is a violating pair in the
transitive closure.  Both testers use these violating rectangles, and the disjointness of their
left and right sides allows their detection probabilities to be combined across buckets.

\paragraph{Combining edge and vertex sampling.}
Before making any queries, the first tester samples
edges from $E$ and independently draws two random multisets of vertices $Q_L,Q_R$; it then checks
every reachable pair between them
(\Cref{alg:edge-vertex-sampling}); reachability computation is free in our model.  Fix a threshold $\tau$.
If $|F_f|\ge\tau$, direct sampling from $E$ succeeds with a constant probability using
$q=O(m/\tau)$ samples.  If $|F_f|<\tau$, the tester rejects whenever, for some bucket $e$,
$Q_L$ meets $A_e^+$ and $Q_R$ meets $B_e^-$.

In the second case, a birthday-paradox argument for hitting both sides of the same violating
rectangle shows that $q=O(\sqrt{\tau}/\varepsilon)$ samples in each of $Q_L$ and $Q_R$ suffice.
Balancing the two requirements by taking
$\tau=(\varepsilon m)^{2/3}$ yields
\[
  q=O\!\left(\frac{m^{1/3}}{\varepsilon^{2/3}}\right)
\]
queries (\Cref{thm:m13}), independently of $\ell$.

\paragraph{Sampling from the transitive reduction and closure.}
The second tester also uses these violating rectangles but samples pairs directly from the transitive
closure instead of combining vertices drawn independently for $Q_L$ and $Q_R$
(\Cref{alg:tr-hybrid}).  This makes the
detection argument even simpler.  When $|F_f|<\tau$, the violating rectangles certify that the closure
contains at least
\[
  \frac14\sum_{e\in F_f}|M_e|^2
  \ge \frac{|M|^2}{4|F_f|}
  > \frac{\varepsilon^2n^2}{4\tau}
\]
violating pairs.  Hence $q=O(\tau\ell/(\varepsilon^2n^2))$ samples from the transitive closure
suffice in this case, whereas $q=O(m/\tau)$ samples of edges from the transitive reduction suffice
when $|F_f|\ge\tau$.
Balancing them at $\tau=\varepsilon n\sqrt{m/\ell}$ gives
\[
  q=O\!\left(\frac{\sqrt{m\ell}}{\varepsilon n}\right)
\]
queries (\Cref{thm:sqrt-ml}).  For constant $\varepsilon$, this improves the classical
$O(\sqrt n)$ bound exactly when $m\ell=o(n^3)$.

\subsubsection{A lower-bound framework via positive matchings}
\paragraph{From induced matchings to positive matchings.}
The classical approach to proving lower bounds for monotonicity testing on general posets, pioneered by
Fischer et al.~\cite{fischer2002monotonicity}, is built around
\emph{Ruzsa--Szemer\'edi (RS) graphs}: one packs many large \emph{induced} matchings and uses a
randomly selected matching to generate YES/NO instances that are difficult to distinguish.  Matchings
are natural because their edges are vertex-disjoint.  If many edges in the matching are violated,
then repairing the function requires changing
at least one endpoint of each edge; at the same time, a tester typically needs to query both endpoints
of a matched pair to see the correlation that distinguishes YES from NO.

The inducedness requirement plays two opposing roles in the construction for Boolean-valued
functions.  On the one hand, it keeps the violation pattern clean, since extra edges among the
endpoints of matching edges would immediately introduce additional Boolean violations.  On the
other hand, the strength of the RS framework is governed by how many large induced matchings can
be packed into one graph: the hard distribution hides a uniformly selected matching among the
$t$ choices.  The construction of Fischer et al.~\cite{fischer2002monotonicity}, however, packs
only $t=n^{\Theta(1/\log\log n)}$ induced matchings of linear size and consequently yields a
lower bound of $n^{\Omega(1/\log\log n)}$ queries, still far below the birthday bound $\sqrt n$.

Our first step is to identify the correct replacement for inducedness when the function range is
$\mathbb R$.  On a bipartite DAG, the question ``can a matching be realized as exactly the
set of violating edges?'' has a precise answer: this is possible if and only if the matching is
positive (\Cref{lem:pm_equiv}).  Informally, a matching $M$ is positive if there is an assignment
$w$ of weights to vertices such that
\[
  w(\ell)+w(r)>0 \quad\Longleftrightarrow\quad (\ell,r)\in M
\]
on every edge of the bipartite graph under consideration.  Equivalently, positivity excludes alternating
closed walks on the vertices saturated by $M$ (\Cref{thm:positive_alt_walk}).  This motivates
PMRS families, which replace induced matchings by large positive matchings with pairwise disjoint
edge sets (\Cref{def:pmrs}).  This relaxation is unavailable for the Boolean range, but for
real-valued functions it bypasses the limitations on packing induced matchings while retaining
exactly the structural property required by the lower-bound argument.

\paragraph{From PMRS to a lower bound with two-sided error.}
Suppose a bipartite graph contains $s$ positive matchings with pairwise disjoint edge sets,
$M_1,\ldots,M_s$, each containing $\Theta(n)$ edges.  We hide a uniform index $I\in[s]$ and use positivity to
construct a monotone function $g_I$ that is equal on the endpoints of $M_I$ and has constant
slack on every edge outside $M_I$.  We then add independent noise from an alphabet of constant size
to the endpoints of the edges in $M_I$.  In the YES distribution, the two endpoints of each edge in $M_I$ receive
the same noise.  In the NO distribution, the noise at the right endpoint is cyclically shifted, so
a constant fraction of the edges in $M_I$ are violated while the slack prevents violations outside $M_I$
(\Cref{sec:pmrs-twosided-lb}).

The key fact is that the YES and NO distributions have the same marginal distribution at every
vertex; they differ only in the local correlation between the two endpoints
of an edge in $M_I$.  Hence, for any query set $Q$ chosen non-adaptively, the transcript has the
same distribution in the two cases unless $Q$ contains both endpoints of an edge in the randomly selected matching.  Since the
matchings are edge-disjoint, $Q$ contains both endpoints of edges from at most
$|Q\cap L|\,|Q\cap R|\le |Q|^2/4$ of the matchings.
Let $\Tr(\mathcal D,Q)$ denote the transcript on $Q$ when the input is drawn from $\mathcal D$.
Consequently,
\[
  \TV\bigl(\Tr(\mathcal D^+,Q),\Tr(\mathcal D^-,Q)\bigr)
  \le \frac{|Q|^2}{4s},
\]
which yields the $\Omega(\sqrt s)$ lower bound against non-adaptive testers with two-sided error
(\Cref{thm:twosided-lb}).  This controls the distribution of the entire transcript, not merely the
probability of querying a violating edge: the full numerical transcripts are statistically
indistinguishable before both endpoints of an edge in the randomly selected matching are queried.

Positivity is also closed under taking submatchings (\Cref{lem:positive_submatching}).  We may
therefore partition every $M_i$, which contains $\Theta(n)$ edges, into
$\Theta(1/\varepsilon)$ positive matchings of size $\Theta(\varepsilon n)$, increasing by the same
factor the number of matchings available as choices for the random index.  Applying the
$\Omega(\sqrt s)$ reduction to the refined family gives the additional
$1/\sqrt\varepsilon$ factor (\Cref{lem:pmrs_refinement} and \Cref{cor:pmrs_eps_lb}).

\paragraph{An explicit construction with polynomially many positive matchings.}
It remains to build many positive matchings.  They must be large and have pairwise disjoint edge
sets.  Our construction
(\Cref{sec:explicit_pmrs}) labels each vertex by $(x,z)$, where $x\in[N]^k$ and
$z\in[N^2]$.  A shift vector $a$ indexes the matching
\[
  (x,z)_L\longmapsto(x+a,z+\|a\|_2^2)_R.
\]
The nontrivial point is positivity.  For every $a$ we give an explicit weight function $w_a$ such
that, on an edge belonging to shift $b$,
\[
  w_a(\ell)+w_a(r)=\frac12-\|b-a\|_2^2.
\]
Because the shifts are integral, this expression is positive precisely for $b=a$ and is at most
$-1/2$ otherwise.  Thus the matching associated with every shift is positive, even though it need
not be induced.

The two sides have size $n_0=N^{k+2}$, while the construction contains
$s=\Theta(N^k)=\Theta(n_0^{k/(k+2)})$ positive matchings, each containing $\Theta(n_0)$ edges
(\Cref{thm:explicit_pmrs}).  Applying the reduction above therefore gives
\[
  \Omega\!\left(n_0^{k/(2k+4)}\right)
  =\Omega\!\left(n_0^{1/2-1/(k+2)}\right)
\]
queries, and the refinement above supplies the factor $1/\sqrt\varepsilon$.
Taking $k$ to be a sufficiently large constant gives the final
$\Omega_\delta(n^{1/2-\delta}/\sqrt\varepsilon)$ lower bound
(\Cref{thm:pmrs_near_sqrt}).

The bounded-label-exposure construction described below can also recover a
near-$\sqrt{n/\varepsilon}$ lower bound.  We nevertheless prove this bound first using the simpler
PMRS construction above as a warm-up application of the framework, thereby isolating the basic PMRS
and correlated-noise argument from the additional machinery needed to control $m$ and $\ell$.

\subsubsection{Matching lower bounds for the phase diagram}
\paragraph{Why bounded label exposure is needed.}
On sparse graphs, our new testers improve on the classical $\sqrt n$ upper bound.  To prove matching
lower bounds, we strengthen the basic PMRS argument used for the near-$\sqrt{n/\varepsilon}$ lower
bound.  That argument bounds the probability of querying both endpoints of an edge in the randomly
selected matching by $O(|Q|^2/s)$, yielding only an $\Omega(\sqrt s)$ query lower bound.  This
square-root loss suffices in the near-$\sqrt{n/\varepsilon}$ setting but is too costly for the
sparse-graph lower bounds.  We therefore require the PMRS family to have bounded label exposure: for
every query set $Q$ of size at most
$q_0$, only a small fraction of the indices $i\in[s]$ have an edge of $M_i$ whose two endpoints lie
in $Q$ (\Cref{def:bounded-label-exposure-pmrs}).  The same YES and NO distributions described above then give
small total variation for every such $Q$, so every non-adaptive tester with two-sided error needs more than
$q_0$ queries (\Cref{lem:bounded-label-exposure-to-twosided}).

\paragraph{Independent perturbations ensure bounded label exposure.}
To construct such a family of matchings, take shifts
$A=\{0,\ldots,P-1\}^t\times\{0\}^{r-t}$, so $s=|A|=P^t$.  For each shift $a$, use the following
independently perturbed displacement in the $z$-coordinate:
\[
  q(a)=C\|a\|_2^2+\xi_a,
  \qquad C=4s,
  \qquad \xi_a\sim\Unif\{0,\ldots,s-1\}.
\]
The variables $(\xi_a)_{a\in A}$ are mutually independent.
The matching indexed by $a$ again maps
$(x,z)_L$ to $(x+a,z+q(a))_R$.  The large quadratic term preserves positivity: with
$p_a=2Ca$, for every $b\ne a$,
\[
  q(b)-q(a)-p_a\cdot(b-a)
  =C\|b-a\|_2^2+\xi_b-\xi_a\ge1.
\]
The resulting sides have size $n_0=\Theta_{r,t}(P^{r+t+2})$, every matching has size
$\Theta(n_0)$, and $U$ has $m_0=\Theta(n_0s)$ edges.

The perturbations ensure bounded label exposure.  Fix a query set $Q\subseteq L\cup R$, and write
$X:=Q\cap L$ and $Y:=Q\cap R$.  Let $\Expose(Q)$ be the set of
indices $a$ for which $Q$ contains both endpoints of an edge in the matching indexed by $a$.  Let
$r_a(Q)$ count pairs in $X\times Y$ for which the right endpoint's $x$-coordinate minus
the left endpoint's $x$-coordinate is $a$.  The event $a\in\Expose(Q)$ can occur only if, for one such pair, the
difference of the $z$-coordinates equals $q(a)$, or equivalently its value after subtracting
$C\|a\|_2^2$ equals the independent random variable $\xi_a$.  Hence
\[
  \Pr_{\xi_a}[a\in\Expose(Q)]\le\frac{r_a(Q)}s.
\]
Since $|\Expose(Q)|=\sum_{a\in A}\mathbf 1[a\in\Expose(Q)]$, linearity of expectation gives
\[
  \mathbb E_{\xi}|\Expose(Q)|
  =\sum_{a\in A}\Pr_{\xi_a}[a\in\Expose(Q)]
  \le\frac1s\sum_{a\in A}r_a(Q)
  \le\frac{|X||Y|}{s}
  \le\frac{|Q|^2}{4s}.
\]
The indicators of the events $a\in\Expose(Q)$ are independent over $a$.  Thus, a Chernoff bound for
each fixed $Q$, followed by a union bound over all query sets, gives a single choice of the perturbations
for which $|\Expose(Q)|\le s/20$ whenever $|Q|\le q_0$, with
\[
  q_0=\Theta_{r,t}\!\left(\frac{s}{\log n_0}\right)
\]
(\Cref{lem:random-height-pmrs}).  Together with
\Cref{lem:bounded-label-exposure-to-twosided}, this shows that $\Omega_{r,t}(s/\log n_0)$ queries are necessary.

\paragraph{Planting a NO instance in a uniformly random copy.}
A PMRS family with bounded label exposure gives the required lower bound when the distance parameter is
constant.  To recover the dependence on $\varepsilon$, take $T=\Theta(1/\varepsilon)$ disjoint
copies and plant one NO instance in a uniformly random copy.  If a query set uses fewer than $Tq_0$
queries, it cannot spend more than $q_0$ queries on more than a small fraction of the copies.
Conditioned on selecting a copy that receives at most $q_0$ queries, the transcripts from that copy
remain indistinguishable.  This gives an $\Omega(Tq_0)$ lower bound
(\Cref{lem:copy-planting}).

For the choice $r=t=k$, the parameters of the graph are
\[
  n_0=\Theta_k(P^{2k+2}),
  \qquad s=P^k,
  \qquad m_0=\Theta_k(P^{3k+2}).
\]
Writing $\alpha_k=k/(3k+2)$, we have
$q_0=\widetilde\Omega(m_0^{\alpha_k})$.  Since the disjoint union of the $T$ copies has $m=Tm_0$
edges,
planting the NO instance in a uniformly random copy gives
\[
  \widetilde\Omega\!\left(m^{\alpha_k}T^{1-\alpha_k}\right).
\]
As $k\to\infty$, $\alpha_k\to1/3$ and $1-\alpha_k\to2/3$.
Absorbing the logarithmic loss into $m^\delta$ yields
$\Omega_\delta(m^{1/3-\delta}/\varepsilon^{2/3})$
(\Cref{thm:m-third-twosided}), matching the first tester up to an arbitrarily small exponent loss.

\paragraph{The cloud lift and the $(m,\ell)$ lower bound.}
To tune $\ell$ separately from $m$, start from the bipartite graph $U$ above, with sides $L$ and $R$.
Attach $D$
new predecessors to every $u\in L$ and $D$ new successors to every $v\in R$, while retaining each
edge of $U$ as a directed edge from $L$ to $R$.  An original vertex together with its attached
vertices forms a cloud $C(x)$, on which the function is constant.  Each original edge $(u,v)$ then
makes all $(D+1)^2$ pairs in $C(u)\times C(v)$ comparable.  Write
$n_{\mathrm{cl}},m_{\mathrm{cl}},\ell_{\mathrm{cl}}$ for the number of vertices, the number of edges
in the transitive reduction, and the number of edges in the transitive closure, respectively.  One
lifted copy with $1\le D\le s$ has
\begin{equation}
\label{eq:cloud-lift-overview-parameters}
  n_{\mathrm{cl}}=\Theta(n_0D),
  \qquad m_{\mathrm{cl}}=\Theta(n_0s),
  \qquad \ell_{\mathrm{cl}}=\Theta(n_0sD^2).
\end{equation}
Replacing each queried vertex by the original vertex of its cloud does not increase the number of
queried vertices, so the bounded-label-exposure guarantee still applies.  At the same time, every violating
edge of $M_i$ supplies $D+1$ vertex-disjoint violating pairs in the lifted graph
(\Cref{lem:cloud-lift-hardness}).

After taking $T=\Theta(1/\varepsilon)$ copies and planting one NO instance in a uniformly random
one, the lower bound is $\Omega(Ts/\log n)$, and \eqref{eq:cloud-lift-overview-parameters} gives
\[
  \frac{\sqrt{m\ell}}n=\Theta(s).
\]
Hence every randomized non-adaptive $\varepsilon$-tester, even with two-sided error, needs
\[
  \Omega\!\left(
    \frac1{\log n}\cdot\frac{\sqrt{m\ell}}{\varepsilon n}
  \right)
\]
queries
(\Cref{thm:ell-dependent-twosided}).  Finally, writing
$s=n_0^\alpha$ and $D=n_0^\beta$ gives
\[
  c=\frac{1+\alpha}{1+\beta},
  \qquad
  d=\frac{1+\alpha+2\beta}{1+\beta}.
\]
Approximating any $0\le\beta\le\alpha<1/2$ by the integer parameters of the construction covers
the entire interior $c+3d<6$ of the red region in \Cref{fig:monotonicity-tester}, up to an
arbitrarily small exponent loss (\Cref{cor:red-region}).

\paragraph{Completing the phase diagram by deterministic padding.}
Adjoin a disjoint DAG whose number of vertices is at most a constant times the number already
present.  Use the same fixed monotone function on this DAG in the YES and NO distributions.  This
preserves the distributions of the observed transcripts, while normalized distance decreases by at most a constant factor
(\Cref{lem:copy-planting-padding}).  For $c\le3/2$ and $c+3d\ge6$, choose the size of each cloud so that the
part built from the PMRS family has $m=n^{c+o(1)}$ and lies near the boundary between the red and blue regions, then
add a chain whose closure has $n^{d+o(1)}$ edges.  For $c\ge3/2$, start from a graph obtained from a PMRS family whose
parameters lie near the point where the three regions meet, add a DAG of height two with
$n^{c+o(1)}$ cover edges, and add a chain with $n^{d+o(1)}$ comparable pairs.  This yields the lower
bound for every fixed pair in \Cref{thm:query-complexity-phase-diagram}\textup{(ii)}.

\subsection{Related work}
\label{sec:related-work}
\paragraph{Monotonicity testing.}
Monotonicity testing on the Boolean hypercube was initiated by Goldreich et al.~\cite{GGLRS00}.
For product domains, subsequent work gave improved testers, near-optimal bounds for hypercubes and
hypergrids, and lower bounds over product domains and hypergrids~\cite{CS13a,CS13b,DGLRRS99,BB16,BRY14,CS14,CDST15,CST14};
connections to isoperimetric inequalities for the Boolean cube were developed in~\cite{KMS15}.
For general posets, testing was initiated by Fischer et al.~\cite{fischer2002monotonicity}, and a
complementary line used transitive-closure spanners and reachability sparsification to design testers
on structured families of graphs~\cite{BGJRW12,Raskhodnikova10}.
Distance estimation and tolerant testing were studied in~\cite{ACCL07,PRR06}; for Boolean functions
on general posets, tolerant-testing consequences also follow from the local-correction and
proper-learning framework of Lange, Rubinfeld, and Vasilyan~\cite{LRV22}.
Local reconstruction of monotone functions and its limitations appear in~\cite{BGJJRW10,SS10}.
For real-valued functions on product domains, monotonicity has been connected to isoperimetric
inequalities for ordered ranges~\cite{BKR23,CS13b}.
\paragraph{Graph-theoretic background.}
Graphs whose edges can be partitioned into induced matchings (Ruzsa--Szemer\'edi graphs) have been
extensively studied since~\cite{RS78}.  In the linear regime $r=cn$, the number of induced matchings
is tightly constrained in several ranges: Fox, Huang, and Sudakov show that it is $O(1)$ for $c>1/4$,
$\Theta(\log n)$ for $c=1/4$, and $O(n/\log n)$ for every fixed $c>1/5$ with $c<1/4$
~\cite{fox_huang_sudakov_induced_2015}.  For every fixed $c<1/4$, Fischer et al. construct graphs with
$n^{\Omega(1/\log\log n)}$ induced matchings of size $cn$~\cite{fischer2002monotonicity}.
If slightly sublinear matchings are allowed, Alon, Moitra, and Sudakov construct nearly complete
graphs whose edges decompose into induced matchings of size $n^{1-o(1)}$~\cite{alon_moitra_sudakov_2011}.

Positive matching decompositions have been studied independently
~\cite{farrokhi_gharakhloo_yazdanpour_pmd_2021}.
Positive matchings are characterized by the absence of alternating closed walks; in bipartite graphs
this is equivalent to being alternating-cycle-free, or uniquely restricted
~\cite{farrokhi_gharakhloo_yazdanpour_pmd_2021,golumbic_hirst_lewenstein_2001}.
For real-valued functions, PMRS families play for positive matchings the role that RS graphs play
for induced matchings.  Bounded label exposure is the additional property that makes the lower bounds in
terms of $m$ and $\ell$ nearly attain the corresponding upper bounds.

\subsection{Organization}
\Cref{sec:preliminaries} introduces notation and basic facts.
\Cref{sec:upper-bounds} establishes the upper bound on query complexity for every fixed pair $(c,d)$.
\Cref{sec:pmrs} then develops the PMRS lower-bound framework, its explicit construction, and the
near-$\sqrt{n/\varepsilon}$ lower bound.
\Cref{sec:bounded-label-exposure-lower-bounds} strengthens the PMRS framework with bounded label exposure and proves a lower bound with the
same exponent for every such pair.
 \section{Preliminaries}
\label{sec:preliminaries}

\subsection{Graph notation}
Let $G=(V,E)$ be a DAG with $n:=|V|$.
We write $u\leadsto v$ if either $u=v$ or there is a directed path of positive length from $u$ to $v$
in $G$.  Thus $\leadsto$ denotes reflexive reachability.  Throughout, transitive closures contain only
nontrivial comparable pairs: the \emph{transitive closure} $\mathrm{TC}(G)$ has edge set
\[
E(\mathrm{TC}(G)) := \{(u,v)\in V\times V : u\neq v\text{ and }u\leadsto v\},
\]
and we write $\ell := |E(\mathrm{TC}(G))|$.
For a DAG $G$, the \emph{transitive reduction} (Hasse diagram) is the unique DAG
$\mathrm{TR}(G)=(V,E_R)$ with $\mathrm{TC}(\mathrm{TR}(G))=\mathrm{TC}(G)$ and minimal edge set.
Since monotonicity depends only on reachability, we may replace $G$ by $\mathrm{TR}(G)$ without loss of
generality and hence assume $G$ is transitively reduced; we write $m:=|E|$.

\subsection{Positive matchings}
Let $\Gamma=(V,E)$ be an undirected graph and let $M\subseteq E$ be a matching.

\begin{definition}[Positive matching~\cite{farrokhi_gharakhloo_yazdanpour_pmd_2021}]\label{def:positive-matching}
The matching $M$ is \emph{positive} (with respect to $\Gamma$) if there exists
a weight function $w:V\to\mathbb{R}$ such that for every edge $e=\{u,v\}\in E$,
\[
w(u)+w(v) > 0 \quad\Longleftrightarrow\quad e\in M,
\]
equivalently: $w(u)+w(v)>0$ for $e\in M$ and $w(u)+w(v)\le 0$ for $e\notin M$.\footnote{Some references require $w(u)+w(v)<0$ for $e\notin M$. For finite graphs, the two formulations are equivalent: a sufficiently small perturbation of a witness $w$ makes all non-matching sums strictly negative without changing the signs on matching edges.}
\end{definition}

We recall alternating walks and the characterization of positive matchings via alternating closed walks,
and its bipartite specialization to alternating-cycle-free (uniquely restricted) matchings.
Let $V(M):=\{v\in V : v \text{ is incident to an edge of } M\}$.

\begin{definition}[Alternating walks and cycles]
Let $\Gamma=(V,E)$ be a graph and $M\subseteq E$ a matching.
An \emph{$M$-alternating walk} is a walk whose edges alternate between $M$ and
$E\setminus M$.
An \emph{$M$-alternating closed walk} is an alternating walk that starts and ends at
the same vertex.
An \emph{$M$-alternating cycle} is an alternating walk that is a (simple) cycle.
\end{definition}

\begin{theorem}[Characterization of positive matchings~\cite{farrokhi_gharakhloo_yazdanpour_pmd_2021,golumbic_hirst_lewenstein_2001}]
\label{thm:positive_alt_walk}
A matching $M$ in $\Gamma$ is positive if and only if the induced subgraph $\Gamma[V(M)]$
contains no $M$-alternating closed walk. In bipartite graphs this is equivalent to the
absence of an $M$-alternating cycle, and such matchings are also known as
\emph{uniquely restricted}.
\end{theorem}

\subsection{Monotonicity and violations}
A function $f:V\to\mathbb{R}$ is \emph{monotone} (order-preserving) if
$f(u)\le f(v)$ for every $(u,v)\in E(\mathrm{TC}(G))$.
Define the distance to monotonicity
\[
d_{\mathrm{mon}}(f) := \min_{g:V\to\mathbb{R}\ \text{monotone}} |\{v\in V:\ f(v)\neq g(v)\}|.
\]
For $\varepsilon>0$, we say that $f$ is \emph{$\varepsilon$-far} from monotonicity if
$d_{\mathrm{mon}}(f)>\varepsilon n$.

A \emph{violating pair} for $f$ is an edge $(u,v)\in E(\mathrm{TC}(G))$ with $f(u)>f(v)$; when
$(u,v)\in E$ we call it a \emph{violating edge}.
The \emph{violation graph} $G_f$ is the bipartite graph with left part $V_L$ (a copy of $V$),
right part $V_R$ (another copy of $V$), and an edge $\{u_L,v_R\}$ whenever $u\leadsto v$ in $G$
and $f(u)>f(v)$.
A set $S\subseteq V$ can be left unchanged by a monotone correction if and only if it is an
antichain in the violation poset defined by $u\prec_f v$ iff $u\leadsto v$ and $f(u)>f(v)$.
Consequently, by Dilworth's theorem and K\H{o}nig's theorem, $d_{\mathrm{mon}}(f)$ equals both the
maximum matching size and the minimum vertex cover size of
$G_f$~\cite{DGLRRS99,fischer2002monotonicity}.
We refer to this as the characterization of distance to monotonicity via matchings in the violation graph.

\subsection{Testers and query complexity}
We consider randomized algorithms with oracle access to $f:V\to\mathbb{R}$.
A tester queries $f$ at selected vertices and, based on the answers to those queries (and its internal randomness),
outputs \emph{accept} or \emph{reject}.
We require \emph{completeness} and \emph{soundness}: for every monotone $f$, the tester accepts with probability at
least $2/3$ (and with probability $1$ for testers with one-sided error), and for every $f$ that is $\varepsilon$-far
from monotone, it rejects with probability at least $2/3$.
A tester is \emph{non-adaptive} if its query set depends only on its internal randomness,
and has \emph{one-sided error} if it always accepts every monotone $f$.
We work in the massively parameterized model in its standard formulation: the input DAG $G$ is given explicitly and only oracle
queries to $f$ are counted. Hence any computation on $G$ (including computing $\mathrm{TR}(G)$,
$\mathrm{TC}(G)$, $\ell$, reachability tests, and sampling uniformly from $E$ or $E(\mathrm{TC}(G))$) is free.
 \section{Upper Bounds Parameterized by the Transitive Reduction and Closure}
\label{sec:upper-bounds}

This section establishes the two new upper bounds underlying the phase diagram.  The testers are
parameterized by the numbers of edges in the transitive reduction and closure; both are non-adaptive
and have one-sided error.  Together with the classical $O(\sqrt{n/\varepsilon})$ tester, they give
the upper bound in \Cref{thm:query-complexity-phase-diagram}\textup{(i)} for every fixed pair of
reachability exponents.  The bounds $O(m^{1/3}/\varepsilon^{2/3})$ and
$O(\sqrt{m\ell}/(\varepsilon n))$ are proved in \Cref{sec:m13-tester,sec:tr-tester}, respectively.

\subsection{\texorpdfstring{An $O(m^{1/3}/\varepsilon^{2/3})$-query tester}{An O(m one-third over epsilon two-thirds)-query tester}}
\label{sec:m13-tester}

In this subsection, we describe a simple tester with one-sided error whose query complexity depends
only on the number of edges in the input DAG and $\varepsilon$. Throughout this subsection, we
assume that the input DAG $G=(V,E)$ is transitively reduced and write $m:=|E|$.

\begin{algorithm}[t]
\caption{Tester combining edge and vertex sampling}
\label{alg:edge-vertex-sampling}
  \textbf{Input:} a transitively reduced DAG $G=(V,E)$ and a proximity parameter $\varepsilon$\;
  $q \gets \Theta\!\left(m^{1/3}/\varepsilon^{2/3}\right)$\;
  Sample a multiset $S_E$ of $q$ edges independently and uniformly from $E$\;
  Independently of each other and of $S_E$, sample two multisets $Q_L,Q_R$ of $q$ vertices uniformly with replacement\;
  Query $f(v)$ for every vertex $v$ in $Q_L\cup Q_R$ or incident to an edge in $S_E$\;
  \lIf{some $(u,v)\in S_E$ satisfies $f(u)>f(v)$}{reject}
  \lIf{there exist $x\in Q_L$ and $y\in Q_R$ such that $(x,y)\in E(\mathrm{TC}(G))$ and $f(x)>f(y)$}{reject}
  Accept\;
\end{algorithm}

\begin{theorem}
\label{thm:m13}
\Cref{alg:edge-vertex-sampling} is a tester for monotonicity with one-sided error and query complexity
\[
O\!\left(m^{1/3}/\varepsilon^{2/3}\right).
\]
\end{theorem}

\begin{proof}
All samples are drawn before any query is made, so \Cref{alg:edge-vertex-sampling} is non-adaptive.
It queries at most $4q$ vertices.
It has one-sided error: if $f$ is monotone, then
there are no violating pairs in $\mathrm{TC}(G)$, and in particular there are no violating edges in $E$,
so the algorithm always accepts.

Fix a function $f:V\to\mathbb{R}$ that is $\varepsilon$-far from monotonicity. By the
characterization via matchings in the violation graph recalled in
\Cref{sec:preliminaries}~\cite{DGLRRS99,fischer2002monotonicity}, there exists a matching $M$ of
violating pairs in $\mathrm{TC}(G)$ such that
\begin{equation}
\label{eq:matching-size}
|M| > \varepsilon n.
\end{equation}
Let
\[
F_f := \{(x,y)\in E : f(x)>f(y)\}
\]
be the set of \emph{violating edges} in $G$.

We use the following dichotomy, governed by a threshold parameter $\tau>0$ to be chosen later.

\paragraph{Case 1: $|F_f|\ge \tau$.}
Each edge in $S_E$ is violating with probability
at least $|F_f|/m \ge \tau/m$. Hence,
\[
\Pr[S_E\cap F_f=\emptyset]
\le \left(1-\frac{\tau}{m}\right)^q
\le \exp\!\left(-q\frac{\tau}{m}\right).
\]
Choosing the hidden constant in $q=\Theta(m/\tau)$ sufficiently large, the right-hand side becomes
at most $1/3$, so $S_E$ contains a violating edge with probability at least $2/3$.

\paragraph{Case 2: $|F_f|< \tau$.}
We analyze rejection based on $Q_L$ and $Q_R$ using the partition obtained by assigning each pair
to its first violating edge, together with a Poisson approximation argument.

For each violating pair $(u,v)\in M$, fix an arbitrary directed path $P_{uv}$ from $u$ to $v$ in $G$.
Since $f(u)>f(v)$, along $P_{uv}$ there must exist at least one edge
$(x,y)\in E$ with $f(x)>f(y)$; define $\psi(u,v)$ to be the \emph{first} such edge
on $P_{uv}$. This yields a partition
\[
M=\biguplus_{e\in F_f} M_e,
\qquad
M_e:=\{(u,v)\in M:\ \psi(u,v)=e\}.
\]
For each $e\in F_f$, let $A_e$ be the set of left endpoints of pairs in $M_e$ and let $B_e$ be the set of
right endpoints. Since $M$ is a matching, the families $\{A_e\}_{e\in F_f}$ and $\{B_e\}_{e\in F_f}$ are
pairwise disjoint, and $|A_e|=|B_e|=|M_e|$.

Fix $e=(x,y)\in F_f$ with $M_e\ne\emptyset$. Let $t_e$ be a median of the multiset $\{f(a):(a,b)\in M_e\}$ and define
\[
A_e^+\ :=\ \{a\in A_e:\ f(a)\ge t_e\},
\qquad
B_e^-\ :=\ \{b\in B_e:\ f(b)< t_e\}.
\]
Then $|A_e^+|,|B_e^-|\ge |M_e|/2$ (since each $(a,b)\in M_e$ has $f(a)>f(b)$).
For any $a\in A_e^+$ and $b\in B_e^-$, let $(a,b_a)$ and $(a_b,b)$ be the unique pairs in $M_e$
containing $a$ and $b$, respectively.  Concatenating the prefix of $P_{a,b_a}$ from $a$ to $x$,
the edge $(x,y)$, and the suffix of $P_{a_b,b}$ from $y$ to $b$ shows that $a\leadsto b$.
Moreover, $f(a)\ge t_e>f(b)$, so $(a,b)\in E(\mathrm{TC}(G))$ is a violating pair.
Consequently, if the sampled multiset $Q_L$ contains at least one vertex in $A_e^+$ and $Q_R$
contains at least one vertex in $B_e^-$ for some $e\in F_f$, then the algorithm rejects.

Let $X_v^L$ (resp., $X_v^R$) be the number of times $v\in V$ is sampled when constructing $Q_L$ (resp., $Q_R$).
We bound the failure event
\[
\mathcal{E} := \bigl[\ \forall e\in F_f,\ \sum_{v\in A_e^+}X_v^L=0\ \ \vee\ \ \sum_{v\in B_e^-}X_v^R=0\ \bigr].
\]
Introduce independent Poisson variables $Y_v^L,Y_v^R\sim \mathrm{Pois}(\lambda)$ with mean $\lambda:=q/n$.
The event $\mathcal E$ is decreasing in each of the two sampled multisets.  Applying the standard
de-Poissonization inequality separately to the samples used to form $Q_L$ and $Q_R$ gives
\begin{equation}
\label{eq:poisson-approx}
\Pr[\mathcal{E}] \le 4\Pr\Bigl[\ \forall e\in F_f,\ \sum_{v\in A_e^+}Y_v^L=0\ \ \vee\ \ \sum_{v\in B_e^-}Y_v^R=0\ \Bigr].
\end{equation}
Indeed, conditional on $\sum_vY_v^L=q$ (respectively, $\sum_vY_v^R=q$), the corresponding
occupancy vector has the multinomial law of $q$ samples drawn independently and uniformly from $V$, and
$\Pr[\mathrm{Pois}(q)\le q]\ge1/2$; see, e.g.,~\cite{barbour1992poisson}.
Since the sets $\{A_e^+\}_e$ are disjoint and the sets $\{B_e^-\}_e$ are disjoint, and the samples
used to form $Q_L$ and $Q_R$ are independent, the events inside the probability in
\eqref{eq:poisson-approx} are independent, and thus
\[
\Pr\Bigl[\ \forall e\in F_f,\ \sum_{v\in A_e^+}Y_v^L=0\ \ \vee\ \ \sum_{v\in B_e^-}Y_v^R=0\ \Bigr]
= \prod_{e\in F_f}\left(1-\bigl(1-e^{-\lambda|A_e^+|}\bigr)\bigl(1-e^{-\lambda|B_e^-|}\bigr)\right).
\]
For each $e\in F_f$, the probability in the Poisson model that both $A_e^+$ and $B_e^-$ are hit is
\[
\bigl(1-e^{-\lambda|A_e^+|}\bigr)\bigl(1-e^{-\lambda|B_e^-|}\bigr)
\ge \bigl(1-e^{-\lambda|M_e|/2}\bigr)^2,
\]
where the inequality follows from $|A_e^+|,|B_e^-|\ge |M_e|/2$.

If there exists $e$ with $\lambda|M_e|\ge 8$, then
$\left(1-e^{-\lambda|M_e|/2}\right)^2\ge (1-e^{-4})^2$ and hence the rejection probability is at least
$1-4(1-(1-e^{-4})^2)>2/3$.

Otherwise, $\lambda|M_e|\le 8$ for all $e\in F_f$. For $x\in[0,4]$ we have $1-e^{-x}\ge x/5$,
so
\[
1-\left(1-e^{-\lambda|M_e|/2}\right)^2 \le 1-\frac{\lambda^2|M_e|^2}{100}.
\]
Therefore,
\[
\Pr[\mathcal{E}]
\le 4\exp\!\left(-\frac{\lambda^2}{100}\sum_{e\in F_f}|M_e|^2\right).
\]
By Cauchy--Schwarz and \eqref{eq:matching-size},
\[
\sum_{e\in F_f}|M_e|^2 \ \ge\ \frac{\left(\sum_{e\in F_f}|M_e|\right)^2}{|F_f|}
\ =\ \frac{|M|^2}{|F_f|}
\ \ge\ \frac{\varepsilon^2 n^2}{|F_f|}
\ >\ \frac{\varepsilon^2 n^2}{\tau},
\]
where we used the assumption $|F_f|<\tau$ in the last inequality.
Plugging this in and using $\lambda=q/n$, we obtain
\[
\Pr[\mathcal{E}] \le 4\exp\!\left(-\frac{\varepsilon^2 q^2}{100\tau}\right).
\]
Thus, by setting $q=\Theta(\sqrt{\tau}/\varepsilon)$ with a sufficiently large hidden constant, the
above probability becomes at most $1/3$, and the algorithm rejects based on $Q_L$ and $Q_R$ with
probability at least $2/3$.

It remains to balance the two cases.  Setting
\[
\tau := (\varepsilon m)^{2/3}
\]
makes $m/\tau$ and $\sqrt{\tau}/\varepsilon$ equal, and both are
$m^{1/3}/\varepsilon^{2/3}$.

In both cases, \Cref{alg:edge-vertex-sampling} rejects $\varepsilon$-far functions with probability at least $2/3$,
and it uses $O(q)=O(m^{1/3}/\varepsilon^{2/3})$ queries.
\end{proof}

\subsection{\texorpdfstring{An $O(\sqrt{m\ell}/(\varepsilon n))$-query tester}{An O(sqrt(m ell)/(epsilon n))-query tester}}
\label{sec:tr-tester}

In this subsection, we give a simple tester with one-sided error whose query complexity depends on
the number of edges in the transitive closure and the number of edges in the transitive reduction.
As in the preceding subsection, the proof assigns each violating pair in a large matching to the
first violating edge on a fixed path.  Direct sampling from the transitive closure then makes the
detection argument simpler.

Throughout this subsection we assume that the input DAG is already transitively reduced and write
$m:=|E|$.

\begin{theorem}
\label{thm:sqrt-ml}
There exists a tester for monotonicity with one-sided error and query complexity
\[
  O\!\left(\frac{\sqrt{m\,\ell}}{\varepsilon n}\right),
\]
where $\ell=|E(\mathrm{TC}(G))|$.
\end{theorem}

\begin{algorithm}[t]
\caption{Tester sampling from the transitive reduction and closure}
\label{alg:tr-hybrid}
  \textbf{Input:} a transitively reduced DAG $G=(V,E)$ and a proximity parameter $\varepsilon$\;
  Set $\ell\gets |E(\mathrm{TC}(G))|$\;
  Set $q \gets \Theta\!\left(\frac{\sqrt{m\,\ell}}{\varepsilon n}\right)$\;
  Sample a multiset $S_E$ of $q$ edges independently and uniformly from $E$\;
  Sample a multiset $S_{\mathrm{TC}}$ of $q$ pairs independently and uniformly from $E(\mathrm{TC}(G))$\;
  Query $f(v)$ for every vertex $v$ incident to a pair in $S_E\cup S_{\mathrm{TC}}$\;
  \lIf{some $(u,v)\in S_E\cup S_{\mathrm{TC}}$ satisfies $f(u)>f(v)$}{reject}
  Accept\;
\end{algorithm}

\begin{proof}
All samples are drawn before any query is made, so the tester is non-adaptive.
It queries at most $4q$ vertices.
It has one-sided error because it rejects only after finding a violating pair.

Fix a function $f:V\to\mathbb{R}$ that is $\varepsilon$-far from monotonicity.
By the characterization via matchings in the violation graph recalled in
\Cref{sec:preliminaries}~\cite{DGLRRS99,fischer2002monotonicity},
there exists a matching $M$ of violating pairs in $\mathrm{TC}(G)$ with
\[
|M| > \varepsilon n.
\]
Let
\[
F_f := \{(x,y)\in E : f(x)>f(y)\}
\]
be the set of \emph{violating edges} in $G$.

We use the following dichotomy, governed by a threshold parameter $\tau>0$ to be chosen later.

\paragraph{Case 1: $|F_f|\ge \tau$.}
Each edge sampled into $S_E$ is violating with
probability at least $\tau/m$.
Thus, by choosing the hidden constant in $q=\Theta(m/\tau)$ large enough, the rejection
probability due to $S_E$ is at least
\[
1-\left(1-\frac{\tau}{m}\right)^q \ge \frac{2}{3}.
\]

\paragraph{Case 2: $|F_f|< \tau$.}
As in Case~2 of the proof of \Cref{thm:m13}, assign each pair in $M$ to the first violating edge on
a fixed path between its endpoints.  This gives a partition
\[
M=\biguplus_{e\in F_f} M_e.
\]
For each $e$, the same argument gives sets $A_e^+$ and $B_e^-$ of left and right endpoints,
respectively, such that
\[
|A_e^+|,|B_e^-|\ge \frac{|M_e|}{2}
\]
and every pair in $A_e^+\times B_e^-$ is violating.  The sets $A_e^+$ are pairwise disjoint across $e$, as are the
sets $B_e^-$.  Hence the rectangles $A_e^+\times B_e^-$ are pairwise disjoint and contain at least
$|M_e|^2/4$ violating pairs each.  Applying Cauchy--Schwarz, the total number of violating pairs in
$\mathrm{TC}(G)$ is therefore at least
\[
\frac{1}{4}\sum_{e\in F_f} |M_e|^2
\ \ge\
\frac{1}{4}\cdot\frac{\left(\sum_{e\in F_f} |M_e|\right)^2}{|F_f|}
\ =\
\frac{|M|^2}{4|F_f|}
\ >\
\frac{\varepsilon^2 n^2}{4\tau}.
\]
Hence, each pair sampled into $S_{\mathrm{TC}}$ is violating with probability at least
$\varepsilon^2 n^2/(4\tau\ell)$.  By choosing the hidden constant in
$q=\Theta(\tau\ell/(\varepsilon^2 n^2))$ large enough, the rejection probability due to
$S_{\mathrm{TC}}$ is at least
\[
1-\left(1-\frac{\varepsilon^2 n^2}{4\tau\ell}\right)^q \ge \frac{2}{3}.
\]

It remains to balance the two cases.
Setting
\[
\tau := \varepsilon n \sqrt{\frac{m}{\ell}}
\]
makes $m/\tau$ and $\tau\ell/(\varepsilon^2 n^2)$ equal up to constants, and thus
\Cref{alg:tr-hybrid} has query complexity
\[
q = \Theta\!\left(\frac{\sqrt{m\,\ell}}{\varepsilon n}\right).
\]
This completes the proof.
\end{proof}

\begin{proof}[Proof of \Cref{thm:query-complexity-phase-diagram}\textup{(i)}]
Fix $c,d,\varepsilon$ and a family of DAGs as in the statement.
Run the least expensive of the classical $O(\sqrt{n/\varepsilon})$-query
tester~\cite{fischer2002monotonicity} and the testers from
\Cref{thm:m13,thm:sqrt-ml}.  Since $\varepsilon$ is fixed, their query complexities are,
respectively,
\[
  n^{1/2+o(1)},
  \qquad n^{c/3+o(1)},
  \qquad n^{(c+d)/2-1+o(1)}.
\]
Taking the minimum proves the claim.
\end{proof}
 \section{\texorpdfstring{A Near-$\sqrt{n/\varepsilon}$ Lower Bound from Positive-Matching RS Families}{A Near-sqrt(n/epsilon) Lower Bound from Positive-Matching RS Families}}\label{sec:pmrs}

In this section, we develop Ruzsa--Szemer\'edi families of positive matchings (PMRS) as a source of
lower bounds for non-adaptive monotonicity testing.  An explicit PMRS construction yields a
near-$\sqrt{n/\varepsilon}$ lower bound even for testers with two-sided error.

We proceed as follows.
We first relate positive matchings to violation patterns in bipartite DAGs and define PMRS families
(\Cref{sec:pmrs-pos-viol,sec:pmrs-analogue-rs}).
We then prove lower bounds from PMRS for non-adaptive testers with two-sided error
(\Cref{sec:pmrs-twosided-lb}), construct PMRS graphs explicitly (\Cref{sec:explicit_pmrs}), and
combine these ingredients to obtain
the main near-$\sqrt{n/\varepsilon}$ lower bound (\Cref{sec:pmrs-main-lb}).

\subsection{Positive matchings and violations}\label{sec:pmrs-pos-viol}
We recall the basic definitions and known characterizations of positive matchings in
\Cref{sec:preliminaries}. Here we relate them to violation patterns in bipartite DAGs.

In the bipartite setting, let $U=(L,R;E)$ be an undirected bipartite graph, and let
$G:=P(U)$ be the bipartite DAG obtained by orienting each edge from $L$ to $R$.

\begin{lemma}[Violations exactly on a matching]
\label{lem:pm_equiv}
Let $U=(L,R;E)$ be bipartite and let $M\subseteq E$ be a matching.
The following are equivalent:
\begin{enumerate}
  \item $M$ is a positive matching in $U$.
  \item There exists a function $f:L\cup R\to\mathbb{R}$ such that for every edge
        $(\ell,r)\in E$,
        \[
        f(\ell) > f(r) \quad\Longleftrightarrow\quad (\ell,r)\in M.
        \]
        In other words, the \emph{violating edges} of $f$ on the bipartite DAG $G$
        are \emph{exactly} the matching edges $M$.
\end{enumerate}
\end{lemma}

\begin{proof}
(1)$\Rightarrow$(2): Let $w$ certify positivity of $M$, i.e.,
$w(\ell)+w(r)>0$ iff $(\ell,r)\in M$ for edges $(\ell,r)\in E$.
Define
\[
f(\ell)=w(\ell)\quad(\ell\in L), \qquad f(r)=-w(r)\quad(r\in R).
\]
Then for $(\ell,r)\in E$,
\[
f(\ell)-f(r)=w(\ell)+w(r),
\]
so $f(\ell)>f(r)$ iff $(\ell,r)\in M$.

(2)$\Rightarrow$(1): Given such $f$, define $w(\ell)=f(\ell)$ for $\ell\in L$
and $w(r)=-f(r)$ for $r\in R$. Then $w(\ell)+w(r)=f(\ell)-f(r)$, hence the sign
pattern on $E$ matches $M$.
\end{proof}

\begin{lemma}[Submatchings preserve positivity]
\label{lem:positive_submatching}
Let $\Gamma=(V,E)$ be a graph and let $M\subseteq E$ be a positive matching in $\Gamma$.
Then every submatching $M'\subseteq M$ is also a positive matching in $\Gamma$.
\end{lemma}
\begin{proof}
Let $M'\subseteq M$ and set $V':=V(M')\subseteq V(M)$.
Suppose for contradiction that $M'$ is not positive.
By \Cref{thm:positive_alt_walk}, the induced subgraph $\Gamma[V']$
contains an $M'$-alternating closed walk $W$.
Since $M$ is a matching and $M'\subseteq M$, no edge of $M\setminus M'$ has an endpoint in $V'$,
and thus $M\cap E(\Gamma[V']) = M'$.
Therefore $W$ is also an $M$-alternating closed walk contained in $\Gamma[V(M)]$,
contradicting the positivity of $M$ (again by \Cref{thm:positive_alt_walk}).
\end{proof}

\subsection{An analogue of RS graphs using positive matchings}\label{sec:pmrs-analogue-rs}
We now formalize PMRS families, a relaxation of RS graphs obtained by replacing induced matchings
with positive matchings. This notion underlies the constructions used below for lower bounds.

\begin{definition}[Ruzsa--Szemer\'edi families of positive matchings (PMRS)]\label{def:pmrs}
Fix $\varepsilon\in(0,1/2]$ and $n_0\in\mathbb{N}$.
A bipartite graph $U=(L,R;E)$ with $|L|=|R|=n_0$ is \emph{$(s,\varepsilon)$-PMRS}
if there exist matchings $M_1,\dots,M_s\subseteq E$ such that:
\begin{enumerate}
  \item (Edge-disjoint) $M_i\cap M_j=\emptyset$ for all $i\neq j$.
  \item (Linear size) $|M_i|\ge \varepsilon n_0$ for all $i$.
  \item (Positive) Each $M_i$ is a positive matching in $U$.
\end{enumerate}
\end{definition}

If $U=(L,R;E)$ is $(s,\varepsilon)$-PMRS with $|L|=|R|=n_0$, then
$
|E|\ge \sum_i |M_i|\ge \varepsilon n_0 s,
$
hence $s\le |E|/(\varepsilon n_0)\le n_0/\varepsilon$.
Our explicit construction combined with refinement attains
$s=n_0^{1-o(1)}/\varepsilon$, approaching this ceiling up to subpolynomial factors.

\subsection{Lower bounds from PMRS for non-adaptive testers with two-sided error}\label{sec:pmrs-twosided-lb}
We begin with a lower bound from PMRS: for any $(s,4\varepsilon)$-PMRS graph, every tester that is
non-adaptive and allowed \emph{two-sided} error requires $\Omega(\sqrt{s})$ queries.
The proof first normalizes the witness for a positive matching to have a constant margin, and then adds
correlated noise to obtain YES/NO distributions.

\subsubsection{Margin and a monotone function with slack}\label{sec:pmrs-margin-slack}
Fix an index $i\in[s]$ and let $M_i$ be the corresponding matching in an $(s,4\varepsilon)$-PMRS graph
$U=(L,R;E)$.
Let $w_i:L\cup R\to\mathbb{R}$ be any weight function certifying that $M_i$ is positive, i.e.,
$w_i(\ell)+w_i(r)>0$ for $(\ell,r)\in M_i$ and $w_i(\ell)+w_i(r)\le 0$ for $(\ell,r)\notin M_i$.
Define
\[
\delta_i:=\min_{e=(\ell,r)\in M_i}\bigl(w_i(\ell)+w_i(r)\bigr).
\]
Since $M_i$ is finite and each term is positive, we have $\delta_i>0$.
Define the normalized witness
\[
\widetilde{w}_i(v):=w_i(v)/\delta_i-\tfrac14 \qquad\text{for all }v\in L\cup R.
\]
Then for every edge $(\ell,r)\in E$,
\begin{equation}\label{eq:margin-w}
\widetilde{w}_i(\ell)+\widetilde{w}_i(r)\ge+\tfrac12 \quad\text{if }(\ell,r)\in M_i,
\qquad
\widetilde{w}_i(\ell)+\widetilde{w}_i(r)\le-\tfrac12 \quad\text{if }(\ell,r)\notin M_i.
\end{equation}
Define $f_i:L\cup R\to\mathbb{R}$ by $f_i(\ell)=\widetilde{w}_i(\ell)$ for $\ell\in L$ and
$f_i(r)=-\widetilde{w}_i(r)$ for $r\in R$. Then, for every $(\ell,r)\in E$,
\[
f_i(\ell)-f_i(r)=\widetilde{w}_i(\ell)+\widetilde{w}_i(r),
\]
so by~\eqref{eq:margin-w} we have $f_i(\ell)-f_i(r)\ge+\tfrac12$ on $M_i$ and $\le -\tfrac12$ on
$E\setminus M_i$.

Let $L_i\subseteq L$ and $R_i\subseteq R$ be the sets of endpoints of $M_i$.
For each matching edge $e=(\ell,r)\in M_i$, define $\Delta_e:=f_i(\ell)-f_i(r)$, so by~\eqref{eq:margin-w}
we have $\Delta_e\ge\tfrac12$.
We convert $f_i$ into a \emph{monotone} function with slack by shifting matching endpoints:
\begin{equation}\label{eq:gi-def}
\begin{aligned}
g_i(\ell) &=
\begin{cases}
f_i(\ell)-\Delta_e/2, & \ell\in L_i,\\
f_i(\ell), & \ell\notin L_i,
\end{cases}\\
g_i(r) &=
\begin{cases}
f_i(r)+\Delta_e/2, & r\in R_i,\\
f_i(r), & r\notin R_i.
\end{cases}
\end{aligned}
\end{equation}
Here, for $\ell\in L_i$ (resp.\ $r\in R_i$), $e$ denotes the unique edge of $M_i$ incident to
$\ell$ (resp.\ $r$); this is well-defined since $M_i$ is a matching.
Then for every edge $(\ell,r)\in E$ we have
\begin{equation}\label{eq:gi-slack}
(\ell,r)\in M_i \implies g_i(\ell)=g_i(r),
\qquad
(\ell,r)\notin M_i \implies g_i(\ell)\le g_i(r)-\tfrac12.
\end{equation}
Indeed, for non-matching edges we have $f_i(\ell)-f_i(r)\le -\tfrac12$ by~\eqref{eq:margin-w}, and the
shifts in~\eqref{eq:gi-def} only decrease the left endpoint (if any) and increase the right endpoint
(if any), so the difference can only decrease.
In particular, $g_i$ is monotone on the bipartite DAG $G$, and every non-matching edge has
slack at least $1/2$.

\subsubsection{YES/NO distributions}
We now define two distributions over functions.
Fix the step size $h:=1/8$ and use the noise alphabet $\{0,1,2,3\}$.
For each matching edge $e=(\ell,r)\in M_i$, sample $Z_e\sim\Unif(\{0,1,2,3\})$
independently, and define the noise value $\eta_e:=h Z_e\in\{0,h,2h,3h\}$.

\smallskip
\noindent\textbf{YES distribution $\mathcal{D}_{+}$.}
Sample $i\sim\Unif([s])$, build $g_i$ as in~\eqref{eq:gi-def}, sample $\{Z_e\}_{e\in M_i}$,
and define $F_i^{+}:L\cup R\to\mathbb{R}$ by
\[
F_i^{+}(v)=g_i(v)\quad\text{for }v\notin L_i\cup R_i,
\qquad
F_i^{+}(\ell)=g_i(\ell)+\eta_e,\quad F_i^{+}(r)=g_i(r)+\eta_e\quad\text{for }e=(\ell,r)\in M_i.
\]

\smallskip
\noindent\textbf{NO distribution $\mathcal{D}_{-}$.}
Sample $i$ and $\{Z_e\}_{e\in M_i}$ as above, and define $F_i^{-}$ by
\[
F_i^{-}(v)=g_i(v)\quad\text{for }v\notin L_i\cup R_i,
\qquad
F_i^{-}(\ell)=g_i(\ell)+\eta_e,\quad F_i^{-}(r)=g_i(r)+\eta'_e\quad\text{for }e=(\ell,r)\in M_i,
\]
where $\eta'_e := h\cdot((Z_e-1)\bmod 4)$ (i.e., we cyclically shift the noise on the right endpoint).

\smallskip
\noindent The values of $g_i$ may reveal information about $i$; the construction does not require
the index itself to remain hidden.  Conditional on $i$, the marginal distribution at each vertex is
the same under $F_i^+$ and $F_i^-$: outside $L_i\cup R_i$ the value is $g_i(v)$ in both, and at a
matching endpoint the noise is uniform over $\{0,h,2h,3h\}$.  The two distributions differ only in
the correlation between the endpoints of an edge in $M_i$.  Consequently, if a
non-adaptive query set does not contain both endpoints of such an edge, its transcript has the same
distribution under $F_i^+$ and $F_i^-$, as formalized below.

\subsubsection{Completeness, soundness, and indistinguishability}

\begin{lemma}[Completeness]\label{lem:twosided-completeness}
Every function in the support of $\mathcal{D}_{+}$ is monotone on $G$.
\end{lemma}
\begin{proof}
Fix $i$ and consider any edge $(\ell,r)\in E$.

If $(\ell,r)\in M_i$, then $g_i(\ell)=g_i(r)$ by~\eqref{eq:gi-slack} and the same noise $\eta_e$
is added to both endpoints, so $F_i^{+}(\ell)=F_i^{+}(r)$.

If $(\ell,r)\notin M_i$, then $g_i(\ell)\le g_i(r)-1/2$ by~\eqref{eq:gi-slack}.
The noise values lie in $[0,3h]$, hence
\[
F_i^{+}(\ell)-F_i^{+}(r)
\le (g_i(\ell)-g_i(r)) + 3h
\le -\tfrac12 + \tfrac{3}{8}
= -\tfrac18 < 0.
\]
Thus $F_i^{+}(\ell)\le F_i^{+}(r)$ for every edge, i.e., $F_i^{+}$ is monotone.
\end{proof}

\begin{lemma}[Many violations in the NO distribution]\label{lem:twosided-soundness-mass}
Fix $i$ and draw $F_i^{-}\sim\mathcal{D}_{-}$.
Then an edge $e=(\ell,r)\in M_i$ is violated by $F_i^{-}$ iff $Z_e\neq 0$.
Consequently, the violating edges contain a matching of size
\[
X_i := |\{e\in M_i : Z_e\neq 0\}|\sim \mathrm{Bin}(|M_i|,3/4).
\]
In particular, if $X_i>2\varepsilon n_0$, then $F_i^{-}$ is $\varepsilon$-far from monotone.
Moreover, if $|M_i|\ge 4\varepsilon n_0$ (e.g., when $U$ is $(s,4\varepsilon)$-PMRS),
then for all sufficiently large $\varepsilon n_0$,
\[
\Pr[X_i>2\varepsilon n_0]\ \ge\ 1-\exp(-\Omega(\varepsilon n_0))\ \ge\ 9/10.
\]
\end{lemma}
\begin{proof}
For $e=(\ell,r)\in M_i$, we have $g_i(\ell)=g_i(r)$, hence
\[
F_i^{-}(\ell)-F_i^{-}(r)=\eta_e-\eta'_e.
\]
If $Z_e\in\{1,2,3\}$, then $\eta_e-\eta'_e=h>0$ so $e$ is violated.
If $Z_e=0$, then $\eta_e-\eta'_e=0-3h<0$ so $e$ is not violated.
Thus the violated edges in $M_i$ are exactly $\{e\in M_i:Z_e\neq 0\}$, which is a submatching of
$M_i$ of size $X_i$.

If $X_i>2\varepsilon n_0$, then
$X_i> \varepsilon|V|$ because $|V|=2n_0$.  The characterization of distance via matchings in the
violation graph from \Cref{sec:preliminaries} therefore implies $\varepsilon$-farness.

Finally, if $|M_i|\ge 4\varepsilon n_0$, then
$\mu:=\mathbb{E}[X_i]=(3/4)|M_i|\ge 3\varepsilon n_0$.
A Chernoff bound gives
\[
\Pr[X_i\le 2\varepsilon n_0]
\le
\Pr\!\left[X_i\le\left(1-\tfrac13\right)\mu\right]
\le \exp\!\left(-\frac{(1/3)^2\mu}{2}\right)
\le \exp\!\left(-\frac{\varepsilon n_0}{6}\right),
\]
which is at most $1/10$ for all sufficiently large $\varepsilon n_0$.
\end{proof}

The next lemma states that the full transcript of oracle answers has exactly the same distribution
under the YES and NO constructions unless the query set contains both endpoints of a hidden matching
edge.
\begin{lemma}[Identical transcripts when no edge of $M_i$ has both endpoints queried]\label{lem:twosided-indist}
Fix $i\in[s]$ and a query set $Q\subseteq L\cup R$ chosen non-adaptively.
If $Q$ does not contain both endpoints of any edge in $M_i$, then the joint distribution of the
answers returned by the oracle on $Q$ is identical under $F_i^{+}$ and $F_i^{-}$.
\end{lemma}
\begin{proof}
Under the assumption, each matching edge contributes to the transcript through \emph{at most one}
queried endpoint.
For any queried endpoint of a matching edge $e$, the observed noise is a uniform element of
$\{0,h,2h,3h\}$ both in $F_i^{+}$ and in $F_i^{-}$ (because $Z_e$ is uniform and the map
$z\mapsto (z-1)\bmod 4$ is a permutation).
Across distinct edges, the noises are independent in both distributions.
Hence the resulting joint distribution of all answers on $Q$ is identical.
\end{proof}

\begin{theorem}[$\Omega(\sqrt{s})$ lower bound for non-adaptive testers with two-sided error]\label{thm:twosided-lb}
Let $\varepsilon\in(0,1/8]$ and let $U=(L,R;E)$ be any $(s,4\varepsilon)$-PMRS graph with $|L|=|R|=n_0$.
Set $G:=P(U)$.
Suppose that $\varepsilon n_0$ is sufficiently large.  Then any $\varepsilon$-tester for monotonicity on the bipartite DAG $G$ that is non-adaptive and
allowed two-sided error must make $q = \Omega(\sqrt{s})$ queries.
\end{theorem}
\begin{proof}
Suppose that a randomized non-adaptive tester $\mathcal A$ makes at most $q$ queries.  Condition on
a fixed value $\omega$ of its internal random seed.  The resulting ordered query list
$Q_\omega=(v_1,\ldots,v_{q_\omega})$ is fixed and has length $q_\omega\le q$.
For $\sigma\in\{+,-\}$, write
\[
\Tr(\mathcal D_\sigma,Q_\omega):=(F^\sigma(v_1),\ldots,F^\sigma(v_{q_\omega})),
\]
where the remaining randomness is over the hidden index and the noise variables in
$\mathcal D_\sigma$.  We also use $Q_\omega$ for the underlying set of queried vertices, and define
\[
E(Q_\omega):=E\cap (Q_\omega\cap L)\times(Q_\omega\cap R),
\]
the set of edges for which both endpoints are queried.

Condition on the choice of $i\in[s]$.
If $E(Q_\omega)\cap M_i=\emptyset$, then $Q_\omega$ does not contain both endpoints of any edge in $M_i$, so by
\Cref{lem:twosided-indist} the transcript has the same distribution under $F_i^{+}$ and $F_i^{-}$.
Therefore, for this fixed seed $\omega$,
\[
\TV\big(\Tr(\mathcal{D}_{+},Q_\omega),\Tr(\mathcal{D}_{-},Q_\omega)\big)
\ \le\
\Pr_{i\sim\Unif([s])}\big[E(Q_\omega)\cap M_i\neq\emptyset\big].
\]
Since the matchings $M_1,\dots,M_s$ are edge-disjoint, each edge in $E(Q_\omega)$ belongs to at most one
matching. Hence the above event can occur for at most $|E(Q_\omega)|$ indices $i$, implying
\[
\TV\big(\Tr(\mathcal{D}_{+},Q_\omega),\Tr(\mathcal{D}_{-},Q_\omega)\big)
\ \le\ \frac{|E(Q_\omega)|}{s}
\ \le\ \frac{|Q_\omega\cap L|\cdot |Q_\omega\cap R|}{s}
\ \le\ \frac{q^2}{4s}.
\]
Let $\Tr_{\mathcal A}(\mathcal D_\sigma)$ denote the full transcript consisting of the internal
random seed and the oracle answers under $\mathcal D_\sigma$.  Averaging over the seed gives
\begin{equation}\label{eq:tv-upper}
\TV\bigl(\Tr_{\mathcal A}(\mathcal D_+),\Tr_{\mathcal A}(\mathcal D_-)\bigr)
\le \frac{q^2}{4s}.
\end{equation}

By completeness and \Cref{lem:twosided-completeness}, the acceptance probability under
$\mathcal D_+$ is at least $2/3$.  By \Cref{lem:twosided-soundness-mass} and soundness, the
acceptance probability under $\mathcal D_-$ is at most
$(9/10)(1/3)+1/10=2/5$.  Thus $\mathcal A$ distinguishes the two distributions with bias at least
\[
  \Pr[\mathcal{A}\text{ accepts }\mathcal{D}_{+}]
  -\Pr[\mathcal{A}\text{ accepts }\mathcal{D}_{-}]
  \ge \frac{2}{3}-\frac{2}{5}
  = \frac{4}{15}.
\]
The output of $\mathcal A$ is a function of its full transcript, so this bias is upper bounded by
the total variation distance in~\eqref{eq:tv-upper}.  Hence
$q^2/(4s)\ge 4/15$, which yields $q=\Omega(\sqrt{s})$.
\end{proof}

\subsection{An explicit construction of PMRS graphs}
\label{sec:explicit_pmrs}

Fix an integer $k\ge 2$ and a parameter $N\in\mathbb{N}$.
We denote vectors in $[N]^k$ by bold letters.
Let $\|\mathbf{a}\|^2=\sum_{i=1}^k a_i^2$ and $\mathbf{a}\cdot \mathbf{x}=\sum_{i=1}^k a_i x_i$.

\subsubsection{Vertex sets}
Let $L$ and $R$ be two \emph{disjoint copies} of $[N]^k\times [N^2]$:
\[
L=\{(\mathbf{x},z)_L:\mathbf{x}\in[N]^k,\ z\in[N^2]\},
\qquad
R=\{(\mathbf{y},t)_R:\mathbf{y}\in[N]^k,\ t\in[N^2]\}.
\]
Then
\[
n_0:=|L|=|R|=N^k\cdot N^2 = N^{k+2}.
\]

\subsubsection{A family of shift vectors}
Fix a constant $\alpha\in(0,1)$ (to be chosen as a function of $\varepsilon$ and $k$)
and set
\[
P:=\lfloor \alpha N\rfloor,
\qquad
A:=\{0,1,\dots,P\}^k\setminus\{\mathbf{0}\},
\qquad
s:=|A|=(P+1)^k-1.
\]

\subsubsection{Matchings and the bipartite graph they form}
For each $\mathbf{a}\in A$, define a set of edges $M_{\mathbf{a}}\subseteq L\times R$ by
\[
M_{\mathbf{a}}
:=
\Bigl\{
\bigl( (\mathbf{x},z)_L,\ (\mathbf{x}+\mathbf{a},\ z+\|\mathbf{a}\|^2)_R \bigr)
\;:\;
\mathbf{x}\in[N]^k,\ z\in[N^2],\ \mathbf{x}+\mathbf{a}\in[N]^k,\ z+\|\mathbf{a}\|^2\in[N^2]
\Bigr\}.
\]
Vertices near the boundary of the box may be incident to fewer shifts, since edges leaving
$[N]^k\times[N^2]$ are omitted.  This causes no issue here: the proof below only uses the
linear size of each matching.
Let
\[
E:=\bigcup_{\mathbf{a}\in A} M_{\mathbf{a}}
\quad\text{and}\quad
U:=(L,R;E).
\]

\begin{theorem}[PMRS families with polynomially many matchings]
\label{thm:explicit_pmrs}
Fix $\varepsilon\in(0,1/2]$ and an integer $k\ge 2$.
There exists a choice of $\alpha=\alpha(\varepsilon,k)>0$ such that for all sufficiently
large $N$ (so that $P\ge 1$), the bipartite graph $U$ above is an $(s,\varepsilon)$-PMRS
with $n_0=N^{k+2}$ and
\[
s=\Omega\!\bigl(n_0^{k/(k+2)}\bigr).
\]
In particular, taking $k=2$ yields $s=\Omega(\sqrt{n_0})$.
\end{theorem}

\begin{proof}
We verify the three PMRS conditions.

\noindent\textbf{(1) Each $M_{\mathbf{a}}$ is a matching.}
Fix $\mathbf{a}\in A$.
Each left vertex $(\mathbf{x},z)_L$ is incident to \emph{at most one} edge of $M_{\mathbf{a}}$,
namely to $(\mathbf{x}+\mathbf{a},z+\|\mathbf{a}\|^2)_R$ if that vertex lies in $R$.
Conversely, each right vertex has at most one preimage under this translation map.
Hence $M_{\mathbf{a}}$ is a matching.

\noindent\textbf{(2) Edge-disjointness.}
Every edge $e\in E$ belongs to a unique $M_{\mathbf{b}}$: indeed, if
\[
e=\bigl((\mathbf{x},z)_L,(\mathbf{y},t)_R\bigr)\in E,
\]
then necessarily $\mathbf{y}-\mathbf{x}=\mathbf{b}$ for some $\mathbf{b}\in A$, and this $\mathbf{b}$
is uniquely determined by the endpoints. Thus the matchings $\{M_{\mathbf{a}}\}_{\mathbf{a}\in A}$
are pairwise edge-disjoint.

\noindent\textbf{(3) Linear size.}
For $\mathbf{a}\in A$,
\[
|M_{\mathbf{a}}|
=
\Bigl(\prod_{i=1}^k (N-a_i)\Bigr)\,(N^2-\|\mathbf{a}\|^2).
\]
Since $0\le a_i\le P\le \alpha N$ and $\|\mathbf{a}\|^2\le kP^2\le k\alpha^2 N^2$, we have
\[
|M_{\mathbf{a}}|
\ge
((1-\alpha)N)^k\cdot (1-k\alpha^2)N^2
=
n_0\,(1-\alpha)^k(1-k\alpha^2).
\]
Choose $\alpha=\alpha(\varepsilon,k)>0$ small enough so that
\[
(1-\alpha)^k(1-k\alpha^2)\ \ge\ \varepsilon,
\]
which is possible because the left-hand side tends to $1$ as $\alpha\to 0$.
Then $|M_{\mathbf{a}}|\ge \varepsilon n_0$ for all $\mathbf{a}\in A$.

\noindent\textbf{(4) Positivity of each $M_{\mathbf{a}}$.}
Fix $\mathbf{a}\in A$. Define a weight function $w_{\mathbf{a}}:L\cup R\to\mathbb{R}$ by
\[
w_{\mathbf{a}}\bigl((\mathbf{x},z)_L\bigr):= z-2\,\mathbf{a}\cdot\mathbf{x},
\qquad
w_{\mathbf{a}}\bigl((\mathbf{y},t)_R\bigr):= -t+2\,\mathbf{a}\cdot\mathbf{y}-\|\mathbf{a}\|^2+\tfrac12.
\]
Consider any edge $e\in E$. By construction, $e$ lies in a unique $M_{\mathbf{b}}$ for some
$\mathbf{b}\in A$, and has the form
\[
e=\bigl( (\mathbf{x},z)_L,\ (\mathbf{x}+\mathbf{b},z+\|\mathbf{b}\|^2)_R \bigr).
\]
A direct calculation gives
\begin{align*}
w_{\mathbf{a}}\bigl((\mathbf{x},z)_L\bigr)+w_{\mathbf{a}}\bigl((\mathbf{x}+\mathbf{b},z+\|\mathbf{b}\|^2)_R\bigr)
&=
\bigl(z-2\mathbf{a}\cdot\mathbf{x}\bigr)
+\bigl(-(z+\|\mathbf{b}\|^2)+2\mathbf{a}\cdot(\mathbf{x}+\mathbf{b})-\|\mathbf{a}\|^2+\tfrac12\bigr)\\
&=
\tfrac12-\|\mathbf{b}-\mathbf{a}\|^2.
\end{align*}
Since $\mathbf{a},\mathbf{b}$ have integer coordinates, $\|\mathbf{b}-\mathbf{a}\|^2=0$ iff $\mathbf{b}=\mathbf{a}$,
and otherwise $\|\mathbf{b}-\mathbf{a}\|^2\ge 1$. Hence
\[
w_{\mathbf{a}}(u)+w_{\mathbf{a}}(v)>0 \iff \mathbf{b}=\mathbf{a} \iff e\in M_{\mathbf{a}},
\]
and for every edge $e\notin M_{\mathbf{a}}$ we have $w_{\mathbf{a}}(u)+w_{\mathbf{a}}(v)\le -\tfrac12\le 0$.
This is exactly the definition of $M_{\mathbf{a}}$ being a positive matching in $U$.

\noindent\textbf{(5) Counting matchings.}
Finally,
\[
s=|A|=(P+1)^k-1=\Theta(N^k).
\]
Since $n_0=N^{k+2}$, this yields
\[
s=\Theta\!\bigl(n_0^{k/(k+2)}\bigr),
\]
as claimed.
\end{proof}

\subsection{Main lower bound}\label{sec:pmrs-main-lb}
We now combine the PMRS lower-bound argument with the explicit construction.
As a first step, we state a refinement lemma that converts larger positive matchings
into many smaller ones while preserving positivity.
\begin{lemma}[Refining a PMRS family to smaller $\varepsilon$]
\label{lem:pmrs_refinement}
Let $\varepsilon_0\in(0,1/2]$, let $U=(L,R;E)$ be an $(s,\varepsilon_0)$-PMRS graph with
$|L|=|R|=n_0$, and fix witnessing matchings $M_1,\ldots,M_s$.
Then for every $\varepsilon\in(0,\varepsilon_0]$, the same graph contains an $(s',\varepsilon)$-PMRS family with
\[
s' \;=\; \sum_{i=1}^s \Big\lfloor \frac{|M_i|}{\lceil\varepsilon n_0\rceil}\Big\rfloor
\;\ge\; s\Big\lfloor \frac{\varepsilon_0 n_0}{\lceil\varepsilon n_0\rceil}\Big\rfloor.
\]
If $\varepsilon n_0\ge1$, then in particular
\[
s'\ge \frac{s\varepsilon_0}{4\varepsilon}.
\]
\end{lemma}
\begin{proof}
The matchings $M_1,\dots,M_s$ are positive and have pairwise disjoint edge sets by
\Cref{def:pmrs}.
Fix $i$ and define
\[
b:=\lceil\varepsilon n_0\rceil,
\qquad
t_i:=\big\lfloor |M_i|/b\big\rfloor.
\]
Partition $M_i$ arbitrarily into $t_i$ pairwise edge-disjoint submatchings $M_{i,1},\dots,M_{i,t_i}$,
each of size exactly $b$ (discarding a remainder of size $<b$).
Each $M_{i,j}$ is a matching, and it is positive by \Cref{lem:positive_submatching}.
Since the original $M_i$'s are edge-disjoint across different $i$, all $M_{i,j}$ are edge-disjoint as well.
Moreover, $|M_i|\ge \varepsilon_0 n_0$ implies
$t_i\ge \lfloor \varepsilon_0 n_0/b\rfloor$.
Thus the collection $\{M_{i,j}\}$ forms an $(s',\varepsilon)$-PMRS family with $s'=\sum_i t_i$
and the first stated lower bound on $s'$.

It remains to prove the final estimate.  Suppose $\varepsilon n_0\ge1$ and set
$x:=\varepsilon_0/\varepsilon\ge1$.  Since $b\le 2\varepsilon n_0$, we have
$|M_i|/b\ge x/2$.  Also, $|M_i|\ge\varepsilon_0n_0\ge\varepsilon n_0$ and integrality give
$|M_i|\ge b$.  Hence $|M_i|/b\ge1$, and therefore
\[
t_i=\left\lfloor\frac{|M_i|}{b}\right\rfloor
\ge \frac{|M_i|}{2b}
\ge \frac{x}{4}.
\]
Summing over $i$ proves $s'\ge s\varepsilon_0/(4\varepsilon)$.
\end{proof}

\begin{corollary}[PMRS lower bound with two-sided error and dependence on $\varepsilon$]
\label{cor:pmrs_eps_lb}
Let $\varepsilon_0\in(0,1/2]$ and let $U$ be an $(s,\varepsilon_0)$-PMRS graph with $|L|=|R|=n_0$.
Set $G:=P(U)$.
Then for every $\varepsilon\in(0,\varepsilon_0/4]$ such that $\varepsilon n_0$ is sufficiently large, every $\varepsilon$-tester for monotonicity on
$G$ that is non-adaptive and allowed two-sided error requires
\[
q \;=\; \Omega\!\left(\sqrt{\frac{s\varepsilon_0}{\varepsilon}}\right)
\]
queries.
\end{corollary}
\begin{proof}
By \Cref{lem:pmrs_refinement} with parameter $4\varepsilon$, $U$ contains an $(s',4\varepsilon)$-PMRS
family with $s'\ge s\varepsilon_0/(16\varepsilon)$.
Applying \Cref{thm:twosided-lb} to this family yields $q=\Omega(\sqrt{s'})$,
which gives the claimed bound.
\end{proof}

\begin{theorem}[Near-$\sqrt{n/\varepsilon}$ lower bounds with two-sided error]
\label{thm:pmrs_near_sqrt}
Fix $\varepsilon\in(0,1/8]$. For every $\delta>0$, there exist infinitely many $n$ and bipartite DAGs on
$n$ vertices for which any $\varepsilon$-tester that is non-adaptive and allowed two-sided error needs
\[
q \;=\; \Omega_\delta\!\left(\frac{n^{1/2-\delta}}{\sqrt{\varepsilon}}\right)
\]
queries.
\end{theorem}

\begin{proof}
Let $\varepsilon_0:=1/2$. Choose an integer $k\ge 2$ so that $1/(k+2)\le \delta$.
By \Cref{thm:explicit_pmrs}, for all sufficiently large $N$ there exists an
$(s_0,\varepsilon_0)$-PMRS graph with
$n_0=N^{k+2}$ and $s_0=\Omega_k\!\bigl(n_0^{k/(k+2)}\bigr)$.
Since $\varepsilon$ is fixed, the condition $\varepsilon n_0$ being sufficiently large holds for all
sufficiently large $N$.  Applying \Cref{cor:pmrs_eps_lb} yields
\[
q=\Omega_k\!\left(\sqrt{\frac{s_0\varepsilon_0}{\varepsilon}}\right)
\;=\;\Omega_k\!\left(\frac{n_0^{k/(2k+4)}}{\sqrt{\varepsilon}}\right)
\;=\;\Omega_k\!\left(\frac{n_0^{1/2-1/(k+2)}}{\sqrt{\varepsilon}}\right).
\]
Now let $n:=|V|=2n_0$. Since $n$ and $n_0$ differ only by a factor of $2$,
this is
\[
q=\Omega_k\!\left(\frac{n^{1/2-1/(k+2)}}{\sqrt{\varepsilon}}\right)
=\Omega_\delta\!\left(\frac{n^{1/2-\delta}}{\sqrt{\varepsilon}}\right).
\]
Because this construction exists for all sufficiently large $N$, it gives infinitely many values of $n$.
\end{proof}
 \section[Parameterized Lower Bounds from PMRS Families with Bounded Label Exposure]{Parameterized Lower Bounds from PMRS Families\\with Bounded Label Exposure}
\label{sec:bounded-label-exposure-lower-bounds}

We next strengthen the PMRS framework from \Cref{sec:pmrs} by requiring bounded label exposure.
Our goal is to establish the parameterized lower bounds in
\Cref{thm:query-complexity-phase-diagram}\textup{(ii)}, matching the exponent in
\Cref{thm:query-complexity-phase-diagram}\textup{(i)} for every fixed pair of reachability exponents.
All lower bounds hold even for randomized non-adaptive testers with two-sided error.

\Cref{subsec:bounded-label-exposure} formulates bounded label exposure and derives the resulting
generic lower bound, and \Cref{subsec:random-height-pmrs} constructs PMRS families satisfying this
property.  \Cref{subsec:random-copy-planting,subsec:cloud-lift} establish the quantitative lower
bounds in terms of $m$ and $\ell$, respectively.  Finally,
\Cref{subsec:phase-diagram-padding} completes the proof throughout the phase diagram.

\subsection{A PMRS family with bounded label exposure}
\label{subsec:bounded-label-exposure}

We first isolate an additional property of PMRS families that allows us to strengthen the generic
lower bound from \Cref{sec:pmrs-twosided-lb} and thereby obtain parameterized lower bounds in terms
of $m$ and $\ell$.
Let $U=(L,R;E)$ be a bipartite graph and let $M_1,\ldots,M_s\subseteq E$ be edge-disjoint matchings.  For a query set $Q\subseteq L\cup R$, define
\[
  E(Q):=E\cap ((Q\cap L)\times (Q\cap R))
\]
and
\[
  \Expose(Q):=\{i\in[s]: E(Q)\cap M_i\neq\emptyset\}.
\]
Thus $i\in\Expose(Q)$ precisely when $Q$ queries both endpoints of at least one edge of the $i$th matching.

\begin{definition}[PMRS family with bounded label exposure]
\label{def:bounded-label-exposure-pmrs}
Fix parameters $s,n_0,q_0\in\mathbb N$ and constants $\eta,\rho>0$.  A tuple
$\mathcal U=(U;M_1,\ldots,M_s)$, where $U=(L,R;E)$ is a bipartite graph with
$|L|=|R|=n_0$ and $M_1,\ldots,M_s\subseteq E$ are matchings, is an
$(s,\eta,q_0,\rho)$-PMRS family with bounded label exposure if
\begin{enumerate}[label=(\roman*)]
  \item the matchings are pairwise edge-disjoint;
  \item $|M_i|\ge \eta n_0$ for every $i\in[s]$;
  \item each $M_i$ is positive in $U$;
  \item for every $Q\subseteq L\cup R$ with $|Q|\le q_0$, we have
  \[
    |\Expose(Q)|\le \rho s.
  \]
\end{enumerate}
\end{definition}

The next lemma reuses the information-theoretic argument from
\Cref{sec:pmrs-twosided-lb}, with the crude bound
$|\Expose(Q)|\le |Q\cap L|\,|Q\cap R|$ replaced by the bounded-label-exposure condition.

\begin{lemma}[A lower bound from bounded label exposure]
\label{lem:bounded-label-exposure-to-twosided}
Let $\mathcal U=(U;M_1,\ldots,M_s)$ be a PMRS family with parameters
$(s,\eta,q_0,\rho)$ as in \Cref{def:bounded-label-exposure-pmrs}, where $U=(L,R;E)$ and
$\rho\le 1/20$.  Let $G=P(U)$ be the bipartite DAG obtained by
orienting every edge from $L$ to $R$, and set
\[
  \varepsilon_\star:=\eta/8.
\]
For all sufficiently large $n_0$, there are two distributions $\mathcal D^+$ and $\mathcal D^-$ over functions on $G$ such that:
\begin{enumerate}[label=(\roman*)]
  \item every function in the support of $\mathcal D^+$ is monotone;
  \item a function drawn from $\mathcal D^-$ is $\varepsilon_\star$-far from monotonicity with probability at least $9/10$;
  \item for every query set $Q\subseteq V(G)$ chosen non-adaptively with $|Q|\le q_0$,
  \[
    \TV\bigl(\Tr(\mathcal D^+,Q),\Tr(\mathcal D^-,Q)\bigr)\le \rho .
  \]
\end{enumerate}
Consequently, every $\varepsilon_\star$-tester for monotonicity on $G$ that is randomized, non-adaptive,
and allowed two-sided error makes more than $q_0$ queries.
\end{lemma}

\begin{proof}
We use the YES/NO distributions from \Cref{sec:pmrs-twosided-lb}.  The hidden index $I$ is
uniform over $[s]$.  Conditional on $I=i$, the construction uses positivity of $M_i$ to build a
monotone function $g_i$ with slack and then adds noise from a constant-size alphabet independently on the edges of
$M_i$.  In the YES distribution, the two endpoints of a matching edge receive the same noise;
in the NO distribution, the right endpoint receives a cyclically shifted copy of that noise.

We recall the three properties of that construction.  First, every function in the support of the YES distribution, denoted $\mathcal D^+$, is monotone.  Second, if $F_i^-$ denotes the NO distribution conditioned on $I=i$, then each edge of $M_i$ is violated independently with probability $3/4$.  Hence, by a Chernoff bound, with probability at least $9/10$ the violating edges contain a submatching of $M_i$ of size at least $|M_i|/2\ge \eta n_0/2$.  Since $|V(G)|=2n_0$, the characterization via matchings in the violation graph implies
\[
  d_{\mathrm{mon}}(f)\ge \eta n_0/2 = (\eta/4)|V(G)| > \varepsilon_\star |V(G)|
\]
for such an $f$.  Thus a function drawn from $\mathcal D^-$ is $\varepsilon_\star$-far with probability at least $9/10$.  The use of $\varepsilon_\star=\eta/8$ is only a safety margin for the strict definition of $\varepsilon$-farness.

Third, fix an index $i$ and a query set $Q$ chosen non-adaptively.  If $Q$ does not contain both endpoints of any edge of
$M_i$, then the transcript of answers on $Q$ has exactly the same distribution under $F_i^+$ and
$F_i^-$.  Indeed, each matching edge then contributes either no queried endpoint or one queried
endpoint; in the latter case the observed noise is uniform over the same constant alphabet in both
distributions, and noises on distinct matching edges remain independent.

Now fix a deterministic algorithm that chooses its queries non-adaptively, let $Q$ be its query set,
and suppose that $|Q|\le q_0$.  Let $\Tr(\mathcal D^+,Q)$ and $\Tr(\mathcal D^-,Q)$ be the
transcripts of the answers.  By the preceding paragraph,
\[
  \TV\bigl(\Tr(\mathcal D^+,Q),\Tr(\mathcal D^-,Q)\bigr)
  \le \Pr_{I\sim\Unif([s])}[I\in\Expose(Q)]
  = \frac{|\Expose(Q)|}{s}
  \le \rho .
\]
The same bound holds for a randomized algorithm whose queries are chosen non-adaptively: condition
on its internal randomness and average, noting that the query set is chosen before any answer from the oracle
is observed.

Suppose, for contradiction, that an $\varepsilon_\star$-tester $A$ is randomized, non-adaptive, and allowed
two-sided error, but uses at most $q_0$ queries.  Completeness gives
\[
  \Pr[A\text{ accepts }f\sim\mathcal D^+]\ge 2/3.
\]
For $\mathcal D^-$, with probability at least $9/10$ the input is $\varepsilon_\star$-far, and on such inputs soundness gives acceptance probability at most $1/3$.  On the remaining probability mass we use the trivial upper bound $1$.  Hence
\[
  \Pr[A\text{ accepts }f\sim\mathcal D^-]
  \le \frac{9}{10}\cdot\frac13 + \frac{1}{10}=\frac25.
\]
Thus $A$ distinguishes $\mathcal D^+$ from $\mathcal D^-$ with bias at least $2/3-2/5=4/15$, contradicting the total-variation bound $\rho\le 1/20$.  Therefore more than $q_0$ queries are necessary.
\end{proof}

\subsection{Constructing PMRS families with random heights}
\label{subsec:random-height-pmrs}

We next construct such families by randomly perturbing the heights in the construction from
\Cref{sec:explicit_pmrs}.  The perturbation preserves the separating inequalities needed for
positivity while preventing any query set of small size from aligning with many
shifts $a\in A$.

\begin{lemma}[Existence of PMRS families with bounded label exposure]
\label{lem:random-height-pmrs}
There is an absolute constant $\eta_0>0$ with the following property.  Fix integers $r\ge t\ge1$.
There exists a constant $c_{r,t}>0$ such that, for all
sufficiently large $P$, there is an $(s,\eta_0,q_0,1/20)$-PMRS family with bounded label exposure
$\mathcal U=(U;M_1,\ldots,M_s)$ in the sense of \Cref{def:bounded-label-exposure-pmrs}, where
$U=(L,R;E)$.  Writing $|L|=|R|=n_0$, we have
\[
  s=P^t,
  \qquad
  q_0=\left\lfloor c_{r,t}\frac{s}{\log n_0}\right\rfloor,
\]
and
\[
  n_0=\Theta_{r,t}(P^{r+t+2}),
  \qquad
  |E|=\Theta_{r,t}(n_0s).
\]
\end{lemma}

\begin{proof}
Let
\[
  A:=\{0,1,\ldots,P-1\}^t\times\{0\}^{r-t}
  \subseteq \{0,1,\ldots,P-1\}^r,
  \qquad s:=|A|=P^t.
\]
Set $H:=s$ and $C:=4H$.  For each $a\in A$, choose independently
\[
  \xi_a\sim \Unif\{0,1,\ldots,H-1\}
\]
and define the height
\[
  q(a):=C\|a\|_2^2+\xi_a.
\]
Let $N:=10tP$ and $Z:=20rHP^2$.  Since $\|a\|_2^2\le rP^2$ and $C=4H$, every height satisfies $0\le q(a)\le 5rHP^2\le Z/4$ for all sufficiently large $P$.

Let $L$ and $R$ be two disjoint copies of $[N]^r\times[Z]$.  For $a\in A$, define
\[
  M_a:=\bigl\{((x,z)_L,(x+a,z+q(a))_R):
    x,x+a\in[N]^r,\ z,z+q(a)\in[Z]
  \bigr\}.
\]
Let $E:=\bigcup_{a\in A}M_a$ and $U=(L,R;E)$.

Each $M_a$ is a matching.  Also the matchings are edge-disjoint, because the differences between the first $r$ coordinates of the endpoints of an edge determine $a$ uniquely.  The size of each matching is linear in $n_0:=|L|=|R|=N^rZ$.  In each of the first $t$ coordinates of $x$, the retained fraction is at least $1-1/(10t)$; the remaining $r-t$ coordinates do not reduce the number of choices; and at least a $3/4$ fraction of the choices of $z$ remain.  Bernoulli's inequality therefore gives
\[
  |M_a|
  \ge \left(1-\frac1{10t}\right)^t\frac34n_0
  \ge \frac9{10}\cdot\frac34n_0
  =:\eta_0n_0,
  \qquad \eta_0=\frac{27}{40}.
\]
This also implies $|E|=\Theta_{r,t}(n_0s)$.  Since $N^rZ=\Theta_{r,t}(P^r\cdot P^t\cdot P^2)$, we have $n_0=\Theta_{r,t}(P^{r+t+2})$.

We verify positivity.  Fix $a\in A$ and set $p_a:=2Ca\in\mathbb R^r$.  For $b\neq a$ in $A$,
\[
  q(b)-q(a)-p_a\cdot(b-a)
  = C\|b-a\|_2^2+\xi_b-
    \xi_a
  \ge C-(H-1)
  \ge 1.
\]
The following weights certify positivity:
\[
  w_a((x,z)_L):=z-p_a\cdot x,
\]
\[
  w_a((y,t)_R):=-t+p_a\cdot y+q(a)-p_a\cdot a+\frac12.
\]
For an edge of label $b$, namely $(y,t)=(x+b,z+q(b))$, we have
\[
  w_a((x,z)_L)+w_a((y,t)_R)
  =\frac12-\bigl(q(b)-q(a)-p_a\cdot(b-a)\bigr),
\]
which is positive for $b=a$ and at most $-1/2$ for $b\neq a$.  Hence $M_a$ is positive.

It remains to prove bounded label exposure for a suitable deterministic choice of the perturbations $\{\xi_a\}_{a\in A}$.  Fix a query set $Q=X\cup Y$, where $X\subseteq L$, $Y\subseteq R$, and $|Q|=q$.  For $a\in A$, let $I_a(Q)$ be the indicator that $a\in\Expose(Q)$.  If $a$ is exposed, then for some $((x,z)_L,(y,t)_R)\in X\times Y$ we have
\[
  y-x=a,
  \qquad
  t-z=q(a)=C\|a\|_2^2+\xi_a.
\]
For fixed $Q$ and $a$, let
\[
  r_a(Q):=\bigl|\{((x,z)_L,(y,t)_R)\in X\times Y: y-x=a\}\bigr|.
\]
Then
\[
  \Pr[I_a(Q)=1]\le \frac{r_a(Q)}{H}.
\]
The random variables $I_a(Q)$ are independent over $a$, because $I_a(Q)$ depends only on $\xi_a$.  Moreover,
\[
  \mu_Q:=\mathbb E\sum_{a\in A}I_a(Q)
  \le \frac{1}{H}\sum_{a\in A}r_a(Q)
  \le \frac{|X||Y|}{H}
  \le \frac{q^2}{4s}.
\]
Let $\rho_0:=1/20$ and set
\[
  q_0:=\left\lfloor c_{r,t}\frac{s}{\log n_0}\right\rfloor
\]
for a constant $c_{r,t}>0$ to be fixed.  For any fixed $Q$ with $|Q|\le q_0$,
\[
  \mu_Q\le O\left(\frac{c_{r,t}^2s}{\log^2 n_0}\right).
\]
A Chernoff bound for independent Bernoulli variables with total mean $\mu_Q$ gives
\[
  \Pr\left[\sum_{a\in A}I_a(Q)\ge \rho_0s\right]
  \le
  \left(\frac{e\mu_Q}{\rho_0s}\right)^{\rho_0s}
  \le \exp(-\Omega_{\rho_0}(s\log\log n_0))
\]
for all sufficiently large $P$.

On the other hand, the number of possible query sets of size at most $q_0$ is at most
\[
  \sum_{j\le q_0}\binom{2n_0}{j}
  \le \exp(O_{r,t}(c_{r,t}s)).
\]
Multiplying the failure probability for a fixed $Q$ by the number of possible query sets gives
\[
  \exp\bigl(O_{r,t}(c_{r,t}s)-\Omega(s\log\log n_0)\bigr)=o(1)
\]
for all sufficiently large $P$.  Therefore, by a union bound, with positive probability over the
choice of the perturbations $\{\xi_a\}$, every $Q\subseteq L\cup R$ of size at most $q_0$ exposes
at most $\rho_0s=s/20$ labels.  Fix such a choice of perturbations.
\end{proof}

\subsection{\texorpdfstring{Planting a NO instance in a uniformly random copy for the $m^{1/3}$ lower bound}{Planting a NO instance in a uniformly random copy for the m one-third lower bound}}
\label{subsec:random-copy-planting}

We now add the dependence on the proximity parameter by sampling a NO instance in one uniformly
random copy among many disjoint copies.

\begin{lemma}[Planting in a uniformly random copy]
\label{lem:copy-planting}
Let $G_0$ be a DAG on $n_{\mathrm{base}}$ vertices, and suppose that there are distributions $\mathcal D^+$ and $\mathcal D^-$ over functions on $G_0$ with the following properties for some $\varepsilon_\star>0$ and $q_0\ge1$:
\begin{enumerate}[label=(\roman*)]
  \item $\mathcal D^+$ is supported on monotone functions;
  \item $f\sim\mathcal D^-$ is $\varepsilon_\star$-far from monotonicity with probability at least $9/10$;
  \item for every query set $Q\subseteq V(G_0)$ chosen non-adaptively with $|Q|\le q_0$,
  \[
    \TV\bigl(\Tr(\mathcal D^+,Q),\Tr(\mathcal D^-,Q)\bigr)
    \le \frac1{20}.
  \]
\end{enumerate}
Let $G^{(T)}$ be the disjoint union of $T$ copies of $G_0$.  For every $\varepsilon>0$ and every integer
\[
  1\le T\le \frac{\varepsilon_\star}{2\varepsilon},
\]
every $\varepsilon$-tester for monotonicity on $G^{(T)}$ that is randomized, non-adaptive, and allowed
two-sided error requires
\[
  \Omega(Tq_0)
\]
queries, where the hidden constant is absolute.
\end{lemma}

\begin{proof}
Use the following YES/NO distributions on $G^{(T)}$.  First fix any monotone function $h$ on $G_0$
and choose a copy $J\sim\Unif([T])$.  Under the YES distribution, put an independent sample from
$\mathcal D^+$ in copy $J$ and put $h$ in every other copy.  Under the NO distribution, put the
corresponding sample from $\mathcal D^-$ in copy $J$ and again put $h$ in every other copy.

Every function in the support of the YES distribution is monotone.  With probability at least
$9/10$, the restriction of a function drawn from the NO distribution to copy $J$ is
$\varepsilon_\star$-far on that copy.  Since the graph is a disjoint union, distances add over
connected components, and therefore the whole function has distance more than
\[
  \varepsilon_\star n_{\mathrm{base}}
  \ge 2\varepsilon T n_{\mathrm{base}}
  = 2\varepsilon |V(G^{(T)})|
  > \varepsilon |V(G^{(T)})|.
\]
Thus a function drawn from the NO distribution is $\varepsilon$-far with probability at least $9/10$.

Fix a deterministic algorithm whose queries are chosen non-adaptively, let $Q$ be its query set in
$G^{(T)}$, and write $Q_j$ for the restriction of $Q$ to copy $j$.  Call copy $j$ heavy if
$|Q_j|>q_0$.  The number of heavy copies is at most $|Q|/q_0$, so
\[
  \Pr[J\text{ is heavy}]
  \le \frac{|Q|}{Tq_0}.
\]
Conditioned on $J$ not being heavy, assumption (iii) bounds by $1/20$ the total variation distance
between the transcripts obtained from copy $J$, and all other copies have identical fixed monotone
values under the YES and NO distributions.  Therefore
\[
  \TV(\Tr(\mathcal D^+_T,Q),\Tr(\mathcal D^-_T,Q))
  \le \frac{|Q|}{Tq_0}+\frac{1}{20}.
\]
If $|Q|\le cTq_0$ for an absolute constant $c>0$ chosen sufficiently small, the right-hand side is
less than $4/15$, while completeness and soundness create a distinguishing gap of at least $4/15$
as in \Cref{lem:bounded-label-exposure-to-twosided}.  This contradiction proves the claim.  For a
randomized algorithm that chooses its queries non-adaptively, condition on the internal seed and
average.
\end{proof}

\begin{lemma}[Planting in a uniformly random copy with deterministic padding]
\label{lem:copy-planting-padding}
Assume the hypotheses of \Cref{lem:copy-planting}.  Fix a constant $A\ge0$, let $H$ be any DAG on at
most $ATn_{\mathrm{base}}$ vertices, and let $h_H$ be any fixed monotone function on $H$.  Let $G$ be
the disjoint union of $G^{(T)}$ and $H$.  For every $\varepsilon>0$ and every integer
\[
  1\le T\le \frac{\varepsilon_\star}{2(A+1)\varepsilon},
\]
every $\varepsilon$-tester for monotonicity on $G$ that is randomized, non-adaptive, and allowed
two-sided error requires $\Omega(Tq_0)$ queries.  The hidden constant is absolute and independent of
$A$.
\end{lemma}

\begin{proof}
Use the YES/NO distributions from the proof of \Cref{lem:copy-planting}, which sample from
$\mathcal D^+$ or $\mathcal D^-$ in a uniformly random copy of $G_0$, and assign $h_H$ to $H$
under both distributions.  Every function in the support of the YES distribution remains monotone.
With probability at least $9/10$, the restriction of a function drawn from the NO distribution to
the selected copy requires changing more than $\varepsilon_\star n_{\mathrm{base}}$ values.  Indeed,
the restriction of any monotone repair on $G$ to the selected component is a monotone repair there,
so changes in the deterministic padding cannot reduce
this requirement.  Since the components are disjoint and
\[
  |V(G)|\le (A+1)Tn_{\mathrm{base}},
\]
the assumed upper bound on $T$ gives
\[
  \varepsilon_\star n_{\mathrm{base}}
  \ge 2\varepsilon(A+1)Tn_{\mathrm{base}}
  \ge 2\varepsilon |V(G)|,
\]
so the extended NO function is $\varepsilon$-far.

Queries in $H$ receive the same deterministic answers under YES and NO.  Thus, if $Q_j$ is the
restriction of the query set to the $j$th copy of $G_0$, the proof of
\Cref{lem:copy-planting} applies unchanged and gives
\[
  \TV(\Tr(\mathcal D^+_{T,H},Q),\Tr(\mathcal D^-_{T,H},Q))
  \le \frac{|Q|}{Tq_0}+\frac1{20}.
\]
The same completeness/soundness separation therefore yields the claimed $\Omega(Tq_0)$ lower bound.
\end{proof}

\begin{theorem}[Lower bound matching the $m^{1/3}$ tester, even with two-sided error]
\label{thm:m-third-twosided}
For every constant $\delta>0$, there exist constants $c_\delta,\varepsilon_\delta>0$ such that, for every
$\varepsilon\in(0,\varepsilon_\delta)$ and infinitely many values of $m$, there is a transitively reduced DAG $G$
with $m$ edges such that every $\varepsilon$-tester for monotonicity on $G$ that is randomized,
non-adaptive, and allowed two-sided error makes at least
\[
  c_\delta\frac{m^{1/3-\delta}}{\varepsilon^{2/3}}
\]
queries.
\end{theorem}

\begin{proof}
Choose an integer $k$ large enough that
\[
  \alpha_k:=\frac{k}{3k+2}\ge \frac13-\frac{\delta}{2}.
\]
Apply \Cref{lem:random-height-pmrs} with $r=t=k$.  Then
\[
  s=P^k,
  \qquad
  n_0=\Theta_k(P^{2k+2}),
  \qquad
  m_0:=|E|=\Theta_k(n_0s)=\Theta_k(P^{3k+2}),
\]
and
\[
  q_0=\Theta_k\left(\frac{s}{\log n_0}\right).
\]
Let $\varepsilon_\star=\eta_0/8$ and set
\[
  T:=\left\lfloor \frac{\varepsilon_\star}{2\varepsilon}\right\rfloor.
\]
For $\varepsilon_\delta>0$ small enough, $T=\Theta_k(1/\varepsilon)$ and $T\ge1$.  Let $G$ be the disjoint union of $T$ copies of the bipartite DAG obtained from this graph.  Since all edges in each copy are directed from one bipartition class to the other, $G$ is transitively reduced and has
\[
  m=Tm_0.
\]
By \Cref{lem:bounded-label-exposure-to-twosided}, the associated bipartite DAG has the required YES/NO
distributions.
Applying \Cref{lem:copy-planting}, every $\varepsilon$-tester that is randomized, non-adaptive, and allowed
two-sided error therefore requires
\[
  \Omega_k(Tq_0)
  =\Omega_k\left(\frac{Ts}{\log n_0}\right)
  =\Omega_k\left(\frac{T\,m_0^{\alpha_k}}{\log n_0}\right).
\]
Using $m_0=m/T$, this becomes
\[
  \Omega_k\left(\frac{m^{\alpha_k}T^{1-\alpha_k}}{\log n_0}\right).
\]
Since $T=\Theta_k(1/\varepsilon)$ and $1-\alpha_k\ge 2/3$, for $\varepsilon\in(0,1)$ this is at least
\[
  \Omega_k\left(\frac{m^{\alpha_k}}{\varepsilon^{2/3}\log n_0}\right).
\]
For every fixed $\varepsilon$, taking $P$ arbitrarily large gives infinitely many $m$, and for all sufficiently large such $m$ we have $\log n_0\le m^{\delta/2}$.  Since $\alpha_k\ge 1/3-\delta/2$, the lower bound is
\[
  \Omega_\delta\left(\frac{m^{1/3-\delta}}{\varepsilon^{2/3}}\right),
\]
as claimed.
\end{proof}

\subsection{\texorpdfstring{A cloud lift and the $\ell$-dependent lower bound}{A cloud lift and the ell-dependent lower bound}}
\label{subsec:cloud-lift}

We now prove a matching lower bound for the tester with query complexity $O(\sqrt{m\ell}/(\varepsilon n))$.
The construction applies a cloud lift to the underlying graph of the same PMRS family.  The
family of matchings has bounded label exposure, and the lift increases the number of comparable pairs while
preserving indistinguishability for every query set of small size.

\begin{definition}[Cloud lift]
\label{def:cloud-lift}
Let $U=(L,R;E)$ be bipartite and let $D\ge1$ be an integer.  The $D$-cloud lift $\mathsf{Cl}_D(U)$ is the four-layer DAG with vertex set
\[
  L^-\cup L\cup R\cup R^+,
\]
where
\[
  L^-:=\{(u,a):u\in L,\ a\in[D]\},
  \qquad
  R^+:=\{(v,a):v\in R,\ a\in[D]\}.
\]
Its edges are
\[
  (u,a)\to u\quad (u\in L,a\in[D]),
\]
\[
  u\to v\quad ((u,v)\in E),
\]
and
\[
  v\to (v,a)
  \quad (v\in R,a\in[D]).
\]
For $x\in L\cup R$, let $C(x)$ be its cloud, namely $\{x\}\cup\{(x,a):a\in[D]\}$, where the auxiliary vertices are in $L^-$ if $x\in L$ and in $R^+$ if $x\in R$.  Let
\[
  \pi:L^-\cup L\cup R\cup R^+\to L\cup R
\]
map each vertex $y$ to the unique $x\in L\cup R$ such that $y\in C(x)$.
\end{definition}

\begin{lemma}[Hardness is preserved by a cloud lift]
\label{lem:cloud-lift-hardness}
Let $\mathcal U=(U;M_1,\ldots,M_s)$ be a PMRS family as in
\Cref{def:bounded-label-exposure-pmrs}, with parameters $(s,\eta,q_0,1/20)$, where
$U=(L,R;E)$, $|L|=|R|=n_0$, and $|E|=\Theta(n_0s)$.  Let
$\widetilde G:=\mathsf{Cl}_D(U)$ for some $1\le D\le s$.  Then:
\begin{enumerate}[label=(\roman*)]
  \item $\widetilde G$ is transitively reduced;
  \item with $n_{\mathrm{cl}}:=|V(\widetilde G)|$,
  $m_{\mathrm{cl}}:=|E(\mathrm{TR}(\widetilde G))|$, and
  $\ell_{\mathrm{cl}}:=|E(\mathrm{TC}(\widetilde G))|$,
  \[
    n_{\mathrm{cl}}=2n_0(D+1),
    \qquad
    m_{\mathrm{cl}}=|E|+2n_0D=\Theta(n_0s),
  \]
  and
  \[
    \ell_{\mathrm{cl}}=|E|(D+1)^2+2n_0D=\Theta(n_0sD^2);
  \]
  \item for $\varepsilon_\star:=\eta/8$, there are distributions $\widetilde{\mathcal D}^+$ and $\widetilde{\mathcal D}^-$ over functions on $\widetilde G$ such that $\widetilde{\mathcal D}^+$ is supported on monotone functions, a function drawn from $\widetilde{\mathcal D}^-$ is $\varepsilon_\star$-far with probability at least $9/10$, and every query set $\widetilde Q\subseteq V(\widetilde G)$ chosen non-adaptively with $|\widetilde Q|\le q_0$ satisfies
  \[
    \TV\bigl(\Tr(\widetilde{\mathcal D}^+,\widetilde Q),
              \Tr(\widetilde{\mathcal D}^-,\widetilde Q)\bigr)
    \le \frac1{20}.
  \]
  Consequently, every $\varepsilon_\star$-tester on $\widetilde G$ that is randomized, non-adaptive, and
  allowed two-sided error makes more than $q_0$ queries, for all sufficiently large $n_0$.
\end{enumerate}
\end{lemma}

\begin{proof}
The graph is layered as $L^-\to L\to R\to R^+$.  No edge has an alternative directed path between
the same endpoints.  A vertex $(u,a)\in L^-$ has only one outgoing edge, namely $(u,a)\to u$; an
edge $u\to v$ with $(u,v)\in E$ has no intermediate layer through which a path could pass; and a
vertex $(v,a)\in R^+$ has only one incoming edge, namely $v\to(v,a)$.  Hence the graph is
transitively reduced.

The numbers of vertices and edges in the transitive reduction are immediate from the definition.  For the
transitive closure, every $(u,v)\in E$ contributes all pairs in $C(u)\times C(v)$, namely $(D+1)^2$
comparable pairs.  In addition, the pairs $(u,a)\leadsto u$ and $v\leadsto(v,a)$ contribute
$2n_0D$.  There are no other nontrivial comparable pairs.  Since $|E|=\Theta(n_0s)$ and
$1\le D\le s$, this gives the stated asymptotics.

We now prove the hardness statement.  Use the YES/NO distributions from
\Cref{lem:bounded-label-exposure-to-twosided} on $P(U)$ and pull them back to $\widetilde G$ by defining
\[
  \widetilde F^\pm(x):=F^\pm(\pi(x)).
\]
If $F^+$ is monotone on $P(U)$, then $\widetilde F^+$ is monotone on $\widetilde G$: the endpoints of
$(u,a)\to u$ and $v\to(v,a)$ have equal values, and for every $(u,v)\in E$, all pairs in
$C(u)\times C(v)$ have values $F^+(u)\le F^+(v)$.

For the NO distribution, whenever an edge $e=(u,v)\in M_i$ is violated, every comparable pair in
$C(u)\times C(v)$ is a violating pair in $\widetilde G$.  From each such edge we may choose $D+1$
vertex-disjoint violating pairs, and choices coming from distinct edges of $M_i$ are
vertex-disjoint because $M_i$ is a matching.  As in
\Cref{lem:bounded-label-exposure-to-twosided}, with probability at least $9/10$ at least
$|M_i|/2\ge\eta n_0/2$ edges of $M_i$ are violated.  Hence the violation graph of
$\widetilde F^-$ contains a matching of size at least
\[
  (D+1)\eta n_0/2 = (\eta/4)n_{\mathrm{cl}} > \varepsilon_\star n_{\mathrm{cl}}.
\]
Thus a function drawn from $\widetilde{\mathcal D}^-$ is $\varepsilon_\star$-far with probability at least
$9/10$.

Finally, fix $\widetilde Q\subseteq V(\widetilde G)$ with $|\widetilde Q|\le q_0$ and put
$Q:=\pi(\widetilde Q)\subseteq L\cup R$.  Then $|Q|\le |\widetilde Q|\le q_0$.  If
$I\notin\Expose(Q)$, the transcript on $Q$ has the same distribution under YES and NO.  The transcript
on $\widetilde Q$ is obtained from that on $Q$ by duplicating some coordinates, so it also has the
same distribution under YES and NO.  Therefore
\[
  \TV(\Tr(\widetilde{\mathcal D}^+,\widetilde Q),
       \Tr(\widetilde{\mathcal D}^-,\widetilde Q))
  \le \Pr[I\in\Expose(Q)]
  \le \frac1{20}.
\]
This proves the required total-variation bound on the transcripts.  The final lower bound for one
copy follows from the same completeness/soundness separation as in
\Cref{lem:bounded-label-exposure-to-twosided}.
\end{proof}

\begin{theorem}[Lower bound matching the $\ell$-dependent tester, even with two-sided error]
\label{thm:ell-dependent-twosided}
Fix integers $r\ge t\ge1$.  For all sufficiently large $P$, let $s=P^t$ and let
\[
  n_0=\Theta_{r,t}(P^{r+t+2})
\]
be the size of each bipartition class in the graph from \Cref{lem:random-height-pmrs}.  Let $D$ be any integer with $1\le D\le s$.  Then, for every sufficiently small $\varepsilon>0$, there are transitively reduced DAGs $G$ with parameters
\[
  n=\Theta_{r,t}(Tn_0D),
  \qquad
  m=\Theta_{r,t}(Tn_0s),
  \qquad
  \ell=\Theta_{r,t}(Tn_0sD^2),
\]
where $T=\Theta_{r,t}(1/\varepsilon)$, such that every $\varepsilon$-tester for monotonicity on $G$ that is
randomized, non-adaptive, and allowed two-sided error makes at least
\[
  \Omega_{r,t}\left(
    \frac{1}{\log n}\cdot\frac{\sqrt{m\ell}}{\varepsilon n}
  \right)
\]
queries.
\end{theorem}

\begin{proof}
Start with the PMRS family from \Cref{lem:random-height-pmrs}, take the $D$-cloud lift of its
underlying graph, and then take $T$ disjoint copies, where
\[
  T:=\left\lfloor\frac{\varepsilon_\star}{2\varepsilon}\right\rfloor,
  \qquad
  \varepsilon_\star:=\eta_0/8.
\]
For sufficiently small $\varepsilon$, we have $T=\Theta_{r,t}(1/\varepsilon)$.

By \Cref{lem:cloud-lift-hardness}, one copy of the lifted graph has
\[
  n_{\mathrm{cl}}=\Theta(n_0D),
  \qquad
  m_{\mathrm{cl}}=\Theta(n_0s),
  \qquad
  \ell_{\mathrm{cl}}=\Theta(n_0sD^2).
\]
The YES/NO distributions for one lifted graph satisfy the hypotheses of \Cref{lem:copy-planting} with
\[
  q_0=\Theta_{r,t}(s/\log n_0).
\]
Applying \Cref{lem:copy-planting} gives a lower bound
\[
  \Omega_{r,t}(Tq_0)
  =\Omega_{r,t}\left(\frac{Ts}{\log n_0}\right)
  \ge
  \Omega_{r,t}\left(\frac{Ts}{\log n}\right),
\]
where the last inequality uses $n\ge n_0$.

For the disjoint union of the $T$ copies,
\[
  n=Tn_{\mathrm{cl}},
  \qquad
  m=Tm_{\mathrm{cl}},
  \qquad
  \ell=T\ell_{\mathrm{cl}}.
\]
Therefore
\[
  \frac{\sqrt{m\ell}}{n}
  =
  \frac{\sqrt{(Tm_{\mathrm{cl}})(T\ell_{\mathrm{cl}})}}{Tn_{\mathrm{cl}}}
  =
  \frac{\sqrt{m_{\mathrm{cl}}\ell_{\mathrm{cl}}}}{n_{\mathrm{cl}}}
  =\Theta_{r,t}(s).
\]
Since $T=\Theta_{r,t}(1/\varepsilon)$, we have $Ts=\Theta_{r,t}(\sqrt{m\ell}/(\varepsilon n))$.  Substituting this into the preceding lower bound proves the theorem.
\end{proof}

\begin{corollary}[Exponent form for the red region]
\label{cor:red-region}
Fix constants $c,d$ satisfying
\begin{equation}
\label{eq:red-region-condition}
  1\le c\le d,
  \qquad
  c+3d<6.
\end{equation}
For every $\delta>0$ and every sufficiently small constant $\varepsilon>0$, there are infinitely many $n$ and transitively reduced $n$-vertex DAGs with
\[
  n^{c-\delta}\le m\le n^{c+\delta},
  \qquad
  n^{d-\delta}\le \ell\le n^{d+\delta},
\]
such that every $\varepsilon$-tester that is randomized, non-adaptive, and allowed two-sided error requires
\[
  \Omega_{c,d,\delta}\left(
    n^{-\delta}\cdot\frac{\sqrt{m\ell}}{\varepsilon n}
  \right)
\]
queries.
\end{corollary}

\begin{proof}
Since $\varepsilon$ is fixed, the number $T=\Theta_{r,t}(1/\varepsilon)$ of copies in
\Cref{thm:ell-dependent-twosided} is constant and does not affect the exponents of $n$, $m$, and
$\ell$.  We therefore begin with one lifted copy.  Put
\[
  L:=r+t+2
\]
and take $D=P^b$ for an integer $0\le b\le t$.  Since $n_0=\Theta(P^L)$ and $s=P^t$,
\Cref{thm:ell-dependent-twosided} gives
\[
  n=\Theta(P^{L+b}),
  \qquad
  m=\Theta(P^{L+t}),
  \qquad
  \ell=\Theta(P^{L+t+2b}).
\]
Define the ratios determined by these integer parameters as
\[
  \alpha_L:=\frac{t}{L},
  \qquad
  \beta_L:=\frac{b}{L}.
\]
If $c_L$ and $d_L$ denote the exponents of $m$ and $\ell$ relative to $n$, respectively, then
\begin{equation}
\label{eq:red-exponent-map}
  c_L=\frac{1+\alpha_L}{1+\beta_L},
  \qquad
  d_L=\frac{1+\alpha_L+2\beta_L}{1+\beta_L}.
\end{equation}

The ratios $\alpha_L$ and $\beta_L$ are not chosen independently; they are determined by the
integers $L,t,b$.  To guide the choice of these integers, we first compute the real target ratios
$\alpha_\star,\beta_\star$ that would give the prescribed exponents $c,d$.  By
\eqref{eq:red-exponent-map}, they should satisfy
\begin{equation}
\label{eq:red-target-ratios-system}
  c=\frac{1+\alpha_\star}{1+\beta_\star},
  \qquad
  d=\frac{1+\alpha_\star+2\beta_\star}{1+\beta_\star}.
\end{equation}
Solving \eqref{eq:red-target-ratios-system} gives
\begin{equation}
\label{eq:red-target-ratios}
  \beta_\star=\frac{d-c}{2-d+c},
  \qquad
  \alpha_\star=\frac{c+d-2}{2-d+c}.
\end{equation}
The conditions $1\le c\le d$ and $c+3d<6$ are exactly what we need: they imply
$0\le\beta_\star\le\alpha_\star<1/2$.  Indeed,
\[
  \alpha_\star-\beta_\star=\frac{2(c-1)}{2-d+c}\ge0,
  \qquad
  \frac12-\alpha_\star=\frac{6-c-3d}{2(2-d+c)}>0.
\]

It remains only to approximate $\alpha_\star$ and $\beta_\star$ by valid integer parameters.  Choose a
sufficiently large integer $L$ and set
\[
  t:=\max\{1,\lfloor\alpha_\star L\rfloor\},
  \qquad
  b:=\lfloor\beta_\star L\rfloor,
  \qquad
  r:=L-t-2.
\]
For all sufficiently large $L$, we have $r\ge t\ge1$ and $0\le b\le t$, so these parameters are
valid and $D=P^b\le P^t=s$.  Moreover, $\alpha_L\to\alpha_\star$ and
$\beta_L\to\beta_\star$ as $L\to\infty$.  By \eqref{eq:red-exponent-map}, we can therefore fix $L$
large enough that $|c_L-c|<\delta/2$ and $|d_L-d|<\delta/2$, and then let $P$ tend to infinity.
The constant factors in the estimates above are absorbed by
the remaining $\delta/2$, giving the claimed bounds on $m$ and $\ell$ for infinitely many $n$.
Finally, $1/\log n\ge n^{-\delta}$ for all sufficiently large $n$, so
\Cref{thm:ell-dependent-twosided} gives the stated query lower bound.
\end{proof}

\subsection{Completing the phase diagram by deterministic padding}
\label{subsec:phase-diagram-padding}

We use the color terminology from \Cref{fig:monotonicity-tester}: the red, blue, and yellow regions
are where the minimum defining $u(c,d)$ is attained by $(c+d)/2-1$, $c/3$, and $1/2$,
respectively.  On a boundary two of these terms agree; we include the boundaries in the blue case
below.

The cloud lift above realizes the red interior directly.  We now take instances from the preceding
constructions whose exponent pairs lie on the region boundaries and apply the deterministic padding
lemma from \Cref{lem:copy-planting-padding} to cover the blue and yellow regions as well.

\begin{proof}[Proof of \Cref{thm:query-complexity-phase-diagram}\textup{(ii)}]
Let $\eta_0$ be the absolute constant from \Cref{lem:random-height-pmrs} such that every matching
has size at least $\eta_0n_0$, set $\varepsilon_\star:=\eta_0/8$, and take, for example,
$\varepsilon_0:=\eta_0/128$.  Fix constants $1\le c\le d\le2$, $\delta>0$, and
$\varepsilon\in(0,\varepsilon_0]$.  We prove the stronger statement that there are infinitely many
$n$ and transitively reduced $n$-vertex DAGs with
\[
  n^{c-\delta}\le m\le n^{c+\delta},
  \qquad
  n^{d-\delta}\le \ell\le n^{d+\delta},
\]
on which every randomized non-adaptive $\varepsilon$-tester, even with two-sided error, requires
\[
  \Omega_{c,d,\delta}\!\left(n^{u(c,d)-\delta}\right)
\]
queries.  Fix an auxiliary constant $\gamma>0$ sufficiently small compared with $\delta$.  All
integer parameters below are fixed before the size parameter $P$ tends to infinity, and all
logarithmic factors and constant multiplicative factors can therefore be absorbed into
$n^{\pm\gamma}$ for sufficiently large $P$.

\smallskip
\noindent\emph{Red region: $c+3d<6$.}
Apply \Cref{cor:red-region} with accuracy $\gamma$.  Its proof uses the same absolute
$\eta_0$, so the above choice of $\varepsilon_0$ is valid uniformly over the integer parameters of
that construction; in particular, the number $T$ of copies chosen there is positive and
$\Theta(1/\varepsilon)$ uniformly for $\varepsilon\le\varepsilon_0$.  The resulting graph satisfies the desired bounds on $m$ and $\ell$ once
$\gamma<\delta$.  Moreover,
\[
  \Omega\!\left(
    n^{-\gamma}\frac{\sqrt{m\ell}}{\varepsilon n}
  \right)
  \ge
  \Omega\!\left(
    n^{(c+d)/2-1-2\gamma}
  \right).
\]
In this region $u(c,d)=(c+d)/2-1$, so taking $\gamma$ sufficiently small gives the claimed lower
bound.

\smallskip
\noindent\emph{Blue region and its boundary: $c\le3/2$ and $c+3d\ge6$.}
Set
\[
  \theta:=\frac3c-2\in[0,1].
\]
Choose a sufficiently large integer $k$, and choose $b\in\{0,\ldots,k\}$ so that $b/k$ is as close
to $\theta$ as possible.  Apply \Cref{lem:random-height-pmrs} with $r=t=k$, and take the cloud lift
of the underlying graph with $D=P^b\le s=P^k$.  Write
\[
  A_k:=2k+2+b,
  \qquad M_k:=3k+2.
\]
By \Cref{lem:cloud-lift-hardness}, one lifted copy has
\[
  n_{\mathrm{base}}=\Theta_k(P^{A_k}),
  \qquad
  m_{\mathrm{base}}=\Theta_k(P^{M_k}),
  \qquad
  \ell_{\mathrm{base}}=\Theta_k(P^{M_k+2b}),
\]
and the corresponding YES and NO distributions satisfy the total-variation bound for every query
set of size at most $q_0$, where
\[
  q_0=\Theta_k(P^k/\log P).
\]
Define the three exponents
\[
  c_k:=\frac{M_k}{A_k},
  \qquad
  d_k:=\frac{M_k+2b}{A_k},
  \qquad
  h_k:=\frac{k}{A_k}.
\]
As $k$ tends to infinity along the above choices of $b$,
\[
  c_k\longrightarrow c,
  \qquad
  d_k\longrightarrow 2-\frac c3,
  \qquad
  h_k\longrightarrow\frac c3.
\]
For finite $k$ we also have
\begin{equation}
\label{eq:padding-dk-below-boundary}
  d_k=\frac{6-c_k}{3}-\frac{4}{3A_k}.
\end{equation}
Set
\[
  \widehat d_k:=\max\left\{d,\frac{6-c_k}{3}\right\}.
\]
Because $d\ge(6-c)/3$, we have $\widehat d_k\to d$, while
\eqref{eq:padding-dk-below-boundary} gives
$\widehat d_k>d_k$.

Choose
\[
  T:=\left\lfloor\frac{\varepsilon_\star}{4\varepsilon}\right\rfloor
\]
which is at least one because $\varepsilon\le\varepsilon_0$, and let $N:=Tn_{\mathrm{base}}$ be the
number of vertices in the $T$ copies of the lifted base graph.  Add, as a disjoint component, a transitively reduced
directed chain, containing only its consecutive cover edges, on
\[
  L:=\left\lfloor N^{\widehat d_k/2}\right\rfloor
\]
vertices and assign the constant zero function to this chain under both the YES and NO distributions.  Since
$\widehat d_k\le2$, the chain has at most $N$ vertices for all sufficiently large $N$.
Thus \Cref{lem:copy-planting-padding} applies with padding factor $A=1$ and gives a lower bound
$\Omega(Tq_0)$.  The chain has $\Theta(L)$ edges in its transitive reduction and
$\Theta(L^2)$ edges in its transitive closure.  Since
$\widehat d_k/2\le1\le c_k$ and $\widehat d_k>d_k$, the final graph has
\[
  n=\Theta(N),
  \qquad
  m=\Theta(N^{c_k}),
  \qquad
  \ell=\Theta(N^{\widehat d_k}).
\]
It is transitively reduced because every component is.  Furthermore,
\[
  \Omega(Tq_0)=\Omega_k(TP^k/\log P)=\Omega(N^{h_k-\gamma})
\]
for all sufficiently large $P$.  Taking $k$ large enough makes $c_k$ and $\widehat d_k$ lie within
$\gamma$ of $c$ and $d$, respectively, and makes $h_k\ge c/3-\gamma$.  In the blue region
$u(c,d)=c/3$, which proves the claim, including the red--blue boundary and the point
$(c,d)=(3/2,3/2)$.

\smallskip
\noindent\emph{Yellow region: $c>3/2$.}
Choose a sufficiently large integer $k$ and apply \Cref{lem:random-height-pmrs} with $r=t=k$, without
a cloud lift.  The resulting bipartite DAG has no directed paths of length two, so its transitive
reduction and nontrivial transitive closure have the same edge set.  Thus one copy has
\[
  n_{\mathrm{base}}=\Theta_k(P^{2k+2}),
  \qquad
  m_{\mathrm{base}}=\ell_{\mathrm{base}}=\Theta_k(P^{3k+2}),
  \qquad
  q_0=\Theta_k(P^k/\log P).
\]
By \Cref{lem:bounded-label-exposure-to-twosided}, it has the YES/NO distributions required by
\Cref{lem:copy-planting-padding}.  Take
\[
  T:=\left\lfloor\frac{\varepsilon_\star}{8\varepsilon}\right\rfloor
\]
copies of the base graph; again $T\ge1$ because $\varepsilon\le\varepsilon_0$.  Put
$N:=Tn_{\mathrm{base}}$.  For deterministic padding, first take a height-two bipartite
DAG with two sides of $N$ vertices each and exactly $\lfloor N^c/4\rfloor$ edges directed from the
left side to the right side.  Its transitive reduction and nontrivial transitive closure both consist
of precisely those edges.  Second, take a transitively reduced directed chain, containing only its
consecutive cover edges, on $\lfloor N^{d/2}\rfloor$ vertices.
Assign the constant zero function to both components under YES and NO.  The padding has at most
$3N$ vertices, so \Cref{lem:copy-planting-padding} applies with $A=3$.

For the $T$ copies of the base graph, the exponent of $N$ in their number of edges is
\[
  \frac{3k+2}{2k+2}<\frac32\le c,
\]
the height-two component has $\Theta(N^c)$ edges in both its reduction and closure, and the chain
has $O(N)$ edges in its transitive reduction and $\Theta(N^d)$ edges in its transitive closure.
Since $d\ge c$, the resulting
transitively reduced graph therefore satisfies
\[
  n=\Theta(N),
  \qquad
  m=\Theta(N^c),
  \qquad
  \ell=\Theta(N^d).
\]
The lower bound on the number of queries is
\[
  \Omega(Tq_0)
  =\Omega_k(TP^k/\log P)
  =\Omega\!\left(N^{\frac{k}{2k+2}-\gamma}\right).
\]
Taking $k$ sufficiently large makes the exponent at least $1/2-2\gamma$.  Here
$u(c,d)=1/2$, so this proves the yellow case.  Finally, varying $P$ gives infinitely many values of
$n$ in every case.  Taking the accuracy to zero and then choosing $P$ sufficiently large at each
stage gives the equivalent diagonal formulation in the statement.
\end{proof}
 
\section*{Acknowledgements}
We thank Nathan Harms, Jane Lange, Mikhail Makarov, Cameron Seth, and Yubo Zhang for helpful discussions.
Y.Y. is supported by JSPS KAKENHI Grant Number 22H05001, 25K24465, and 26K21940.

\bibliographystyle{abbrv}
\bibliography{main}

\end{document}